\documentclass[reprint,amsmath,amssymb, aps,prx,twocolumn,superscriptaddress, nofootinbib]{revtex4-2}
\usepackage{graphicx}
\usepackage{hyperref}
\usepackage{amsthm}
\usepackage{dsfont}
\usepackage{xcolor}
\usepackage{soul}
\usepackage{comment}
\usepackage{amssymb}
\usepackage{tikz}
\usepackage[capitalize,nameinlink]{cleveref}
\usepackage{braket}
\usepackage{url}

\usepackage{bm} 
\usepackage{MnSymbol} 
\usepackage[shortlabels]{enumitem}

\usepackage{lineno}

\hypersetup{
    colorlinks = true,
    allcolors = black,
    citecolor = magenta,
    urlcolor = magenta,
}

\newtheorem{theorem}{Theorem}[section]
\newtheorem{corollary}{Corollary}[section]
\newtheorem{lemma}{Lemma}[section]
\newtheorem{proposition}{Proposition}[section]

\theoremstyle{definition}
\newtheorem{definition}{Definition}[section]

\theoremstyle{remark}

\DeclareMathOperator{\Tr}{Tr}

\newcommand{\Trb}[2][]{\Tr_{#1} \left[ #2 \right]}
\newcommand{\Var}[1]{\mathbb{V}_{#1}}
\newcommand{\Varr}{\Var{\rho,H}}
\newcommand{\Loss}{\mathcal{L}_{\rho, H}}
\newcommand{\BH}{\mathcal{B}}
\newcommand{\BHk}{{\mathcal{B}}_\kappa}
\newcommand{\BHz}{{\mathcal{B}}_z}
\newcommand{\deltat}[2]{\Tilde{\delta}_{#1,#2}}

\begin{document}
\preprint{APS/123-QED}
\title{Estimates of loss function concentration in noisy parametrized quantum circuits}

\author{Giulio Crognaletti}
\email{giulio.crognaletti@phd.units.it}
\affiliation{
Department of Physics, University of Trieste, Strada Costiera 11, 34151 Trieste, Italy}
\affiliation{Istituto Nazionale di Fisica Nucleare, Trieste Section, Via Valerio 2, 34127 Trieste, Italy}
\affiliation{European Organization for Nuclear Research (CERN), 1211 Geneva, Switzerland}

\author{Michele Grossi}
\email{michele.grossi@cern.ch}
\affiliation{European Organization for Nuclear Research (CERN), 1211 Geneva, Switzerland}

\author{Angelo Bassi}
\affiliation{
Department of Physics, University of Trieste, Strada Costiera 11, 34151 Trieste, Italy}
\affiliation{Istituto Nazionale di Fisica Nucleare, Trieste Section, Via Valerio 2, 34127 Trieste, Italy}

\date{\today}
\begin{abstract}
Variational quantum computing offers a powerful framework with applications across diverse fields such as quantum chemistry, machine learning, and optimization. However, its scalability is hindered by the exponential concentration of the loss function, known as the barren plateau problem. 
While significant progress has been made and prior work has separately analyzed barren plateaus in unitary and noisy settings, their combined impact remains poorly understood, largely due to limitations in conventional Lie-algebraic approaches.
In this work, we introduce a novel analytical framework based on non-negative matrix theory that enables the description of the variance in layered noisy quantum circuits with arbitrary noise channels. This approach enables the derivation of exact expressions in the deep-circuit regime, uncovering the complex interplay between unitary layers and noise.
Notably, we identify a noise-induced absorption mechanism—a phenomenon absent in purely unitary dynamics—which provides new insight into how noise shapes circuit behavior. We further present a controlled convergence analysis, establishing general lower bounds on the variance of both deep and shallow circuits.
This leads to a principled connection between noise resilience and the expressive capacity of parameterized quantum circuits, particularly under smart initialization strategies. Our theoretical results are supported by numerical simulations and illustrative applications.
\end{abstract}

\maketitle

\section{Introduction}

Quantum computers hold the potential for substantial speed increases across various computational tasks, with a notable example being their capacity to transform our understanding of nature through the simulation of quantum systems \cite{Feynman82, LLoyd96, DiMeglio24}.
In this context, variational quantum computing offers a versatile tool, which combines quantum and classical computational resources. 
The flexibility and relative simplicity of this approach, coupled with its potential for noise resilience, make it a compelling candidate for near-term quantum devices \cite{Fontana21, Vischi24}.

In the past years, a substantial effort has been put in by the community to unlock the potential of variational algorithms. A major hurdle in this direction is the barren plateau (BP) effect, which implies an exponential flattening of the loss landscape, making the optimization step unfeasible with a polynomial amount of quantum resources \cite{McClean18}. More precisely, in the presence of BP, we have an exponentially vanishing probability of being able to efficiently find a loss-minimizing direction after parameter initialization.
In the absence of noise, this phenomenon has been linked to several factors, ranging from the expressive power and entangling capability of the circuit \cite{Larocca22, Marrero21, Sack22, Holmes22, Pesah21, Martin23}, to the locality of $H$ \cite{Cerezo21_2, Uvarov21} and the entanglement of the initial state $\rho$ \cite{Thanasilp23}. Recently, the contributions coming from all such factors were unified in a Lie-algebraic framework \cite{Ragone24, Fontana24, Diaz23}, showing how barren plateaus ultimately arise as a \emph{curse of dimensionality}. Several BP mitigation strategies have been proposed to circumvent the issue. Among the most popular, there are small-angle initializations \cite{Park24, Park24_2, Zhang22, Wang23, Puig24}, which leverage the restriction of the domain $\Theta$ of $\theta$ to limit the expressive power of the circuit, and consequently avoid concentration. \\
In addition to that, the presence of noise is also deemed detrimental, as it often gives rise to noise-induced barren plateau (NIBP)\ \cite{Wang21, Schumann23} and symmetry breaking \cite{Crognaletti24, Tuysuz24}. The former exacerbates the BP effect, producing a deterministic concentration. In this case, an efficient loss-minimizing direction is exponentially hard to estimate, \emph{regardless} of the parameter's initialization choice. These phenomena are mostly linked to decoherence. Ref.~\cite{Larocca24} offers a review of the subject.\\
Indeed, the action of noise can be quite destructive for quantum computation. Along these lines, recent research~\cite{Quek24, Takagi23, Tsubouchi23} puts stringent fundamental bounds on capabilities of quantum error mitigation strategies, emphasizing the limits of scalability of near-term devices, especially in the presence of global depolarizing noise. 
However, such analyses are mostly limited to unital noise, and besides, they typically rely on the properties of circuit ensembles. Indeed, studies have demonstrated that non-unital noise can significantly influence the emergence of NIBP\cite{Mele24, Fefferman23, Singkanipa24}, underscoring that the interplay between general noise and unitary circuit layers in variational quantum computing remains a critical and unresolved area of research. \\
In this work, we approach the problem through a novel formulation grounded in non-negative matrix theory, which enables a unified and general description of these dynamics. This framework not only addresses the previously uncharted aspects of the problem but also seamlessly encompasses both unital and non-unital noise settings. We quantify the deep circuit variance $\Varr^\infty$ for layered circuits composed of local 2-designs interleaved with general noise maps, providing a comprehensive picture. More specifically, we are able to recover known results for unitary circuits and strictly contractive noise maps as limiting cases, while unveiling the emergence of a more complex phenomenon, namely \emph{absorption}, in the intermediate case. Crucially, this can occur only with noise maps that are not strictly contractive—that is, maps that are insufficiently strong to fully erase all quantum information from the system. It is worth emphasizing that this class encompasses both unital and non-unital noise channels, illustrating that unitality of the noise alone is not necessarily detrimental for PQCs. The proposed framework also facilitates a systematic analysis of the noise resilience properties of PQCs, providing a foundation for principled BP mitigation strategies. Notably, we demonstrate a direct correlation between a circuit’s resilience to noise and its capacity for enhanced expressivity under small-angle initialization-like schemes. 

\begin{figure*}
    \centering
    \includegraphics[width=\linewidth]{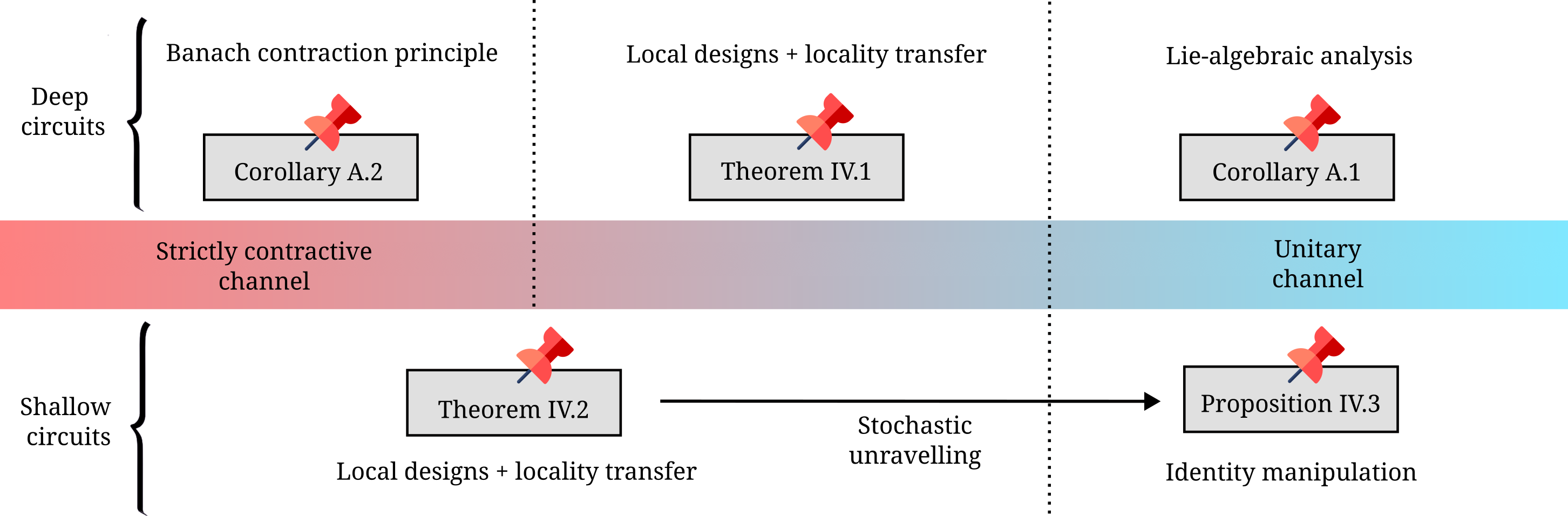}
    \caption{Loss concentration in variational quantum computing.
    The analytical formulation proposed here employs non-negative matrix theory to describe the interplay between local 2-designs and noise. This allows for precise calculation of loss variances for generic noise maps, from strictly contractive to unitary channels, as illustrated by the coloured band in the Figure. 
    The upper part considers deep circuits, where the loss variance $\Varr^L$ reaches its asymptotic limit. While loss concentration is well-understood for strictly contractive \cite{Schumann23, Wang21} and unitary \cite{Ragone24, Fontana24, Diaz23} channels, \cref{thm:main} provides an analytical solution for the intermediate case, revealing a noise-induced absorption mechanism unique to this regime. Consistency with known limiting cases is verified in \cref{sec:methods:recover_limits}.
    The lower part focuses on ``shallow'' circuits, where a lower bound on $\Varr^L$ is established using \cref{thm:lower_bound}. This extends previous works on brickwork circuits \cite{Cerezo21_2, Mele24}, enabling  initialization strategies such as small angle initializations \cite{Zhang22, Wang23} to be represented as stochastic unravellings of noise maps. In \cref{subsect:small_angles} we expand upon their applicability, showing how both unitary and non-unitary QResNets can be derived thanks to \cref{prop:noise_small_angle_equiv}.}
    \label{fig:summary}
\end{figure*}

\section{Motivating examples}\label{sec:intro}

In the following, we consider $n$-qubit quantum systems with a Hilbert space $\mathcal{H}=\bigotimes_{m=1}^M \mathcal{H}_m$ of dimension $d$, divided into $M$ subsystems, each of dimension $d_m$, representing either single qubits or groups of qubits.
More specifically, we study the problem where a quantum state $\rho$ is evolved using a parameterized quantum channel $\Phi_\theta$, whose parameters $\theta = (\theta_1, \theta_2, \dots )\in\Theta$ are optimized by minimizing the loss function
\begin{equation}\label{eq:loss}
\Loss = \Trb{\Phi_\theta(\rho)H},
\end{equation}
where $H$ is an observable of the system.
Here, the presence of BP is diagnosed by studying the variance $\Varr$ of $\Loss$ for varying parameters $\theta$. In particular, we say that $\Loss$ suffers from BP if $\Varr \in O(e^{-\beta n})$, $\beta>0$, as in this case, the loss function exponentially concentrates around its mean value in the number $n$ of qubits.
Within this framework, we focus on the case of \emph{layered} quantum channels, namely
\begin{equation}
\label{eq:circuit_channel}
\Phi_{\theta} = \mathcal{U}_{\theta_{L+1}} \circ \mathcal{E}_{L}\circ\mathcal{U}_{\theta_{L-1}} \dots \circ  \mathcal{E}_{1}\circ\mathcal{U}_{\theta_{1}},
\end{equation}
where each layer is composed of a unitary part $\mathcal{U}_{\theta_l}$ and an arbitrary quantum channel $\mathcal{E}_l$.   
If we denote by $\BH$ the space of bounded operators acting on $\mathcal{H}$, both such components can be regarded as linear, completely positive functions mapping $\BH$ to itself. 
Furthermore, the intermediate channels $\mathcal{E}_l$ will be deemed strictly-contractive if the restriction to the Hermiatian, traceless hyperplane $\mathbb{H}_0 \subset \BH$ of its adjoint action $\mathcal{E}_l^\dagger$ has an induced Schatten 1-norm $\|\mathcal{E}^\dagger_l|_{\mathbb{H}_0}\|_1$ strictly less than unity, $\|\mathcal{E}^\dagger_l|_{\mathbb{H}_0}\|_1<1$, and contractive otherwise, i.e. $\|\mathcal{E}^\dagger_l|_{\mathbb{H}_0}\|_1=1$. According to this definition, contractive channels preserve the $1$-norm of certain operators in $\BH$, while possibly being strictly contractive for others. Moreover, unitary channels, which preserve all norms across their entire domain, are a special case of contractive channels that are never strictly contractive. We refer to \cref{appendix:proof_main} for more details.

Since $\Varr$ is ultimately upper-bounded by the norm of $H$, it is intuitively clear that strictly contractive channels lead to vanishing variance in the deep circuit limit. However, it is important to clarify that this behavior is not inherently linked to the unitality of each individual $\mathcal{E}_l$. Although certain results in the literature may be misconstrued as implying that unital, non-unitary channels inevitably induce deterministic concentration—i.e., NIBP—in deep circuits \cite{Singkanipa24, Wang21}, this interpretation does not hold in general.

A simple yet compelling counterexample involves commuting noise and unitary layers, i.e., when $\mathcal{E}_l \circ \mathcal{U}_{l'}(\rho) = \mathcal{U}_{l'} \circ \mathcal{E}_l(\rho)$ for all $\rho$ and all $l, l'$. In such cases, one can commute all noisy operations to the beginning of the computation, effectively reducing the noisy circuit to a fully unitary one, provided the initial state is redefined as $\tilde{\rho} = \mathcal{E}_L \circ \cdots \circ \mathcal{E}_1(\rho)$. If $\rho$ is a steady state of all intermediate channels $\mathcal{E}_l$, then $\tilde{\rho} = \rho$, and the entire circuit becomes noise-free. Consequently, it cannot exhibit NIBP, although BP may still occur. An explicit example of this is provided in \cref{sec:unital_noise_examples}.

Similarly, even in the non-commuting case, the emergence of concentration is not guaranteed. In fact, the variance can even increase compared to the purely unitary circuit. This counterintuitive behavior can result from noise-induced symmetry breaking \cite{Tuysuz24, Crognaletti24}, a phenomenon typically studied in non-unital channels but also possible in unital ones. In both cases, this effect can lead to an enhanced variance. An example illustrating this mechanism is also presented in \cref{sec:unital_noise_examples}. 

Altogether, the above discussion highlights that, beyond the non-trivial effects observed in non-unital channels, the broader class of contractive, but not \emph{strictly} contractive channels gives rise to a rich variety of largely unexplored variance behaviors. In the following sections we will characterize these scenarios for general channels, encompassing the previously discussed cases as specific instances.


\section{General quantum channels}

To accurately characterize the effects of general noise maps, it is essential to first analyze the structural transformations induced on the computational space by the unitary layers. In particular, it worth noticing that the subdivision of $\mathcal{H}$ into local subsystems induces also a partition of $\BH$. More specifically, if we denote by $\kappa \in \{0,1\}^M$ a binary string of length $M$, we can split $\BH$ into \emph{local} subspaces $\BHk \subset \BH$, each spanned by the \emph{traceless} operators acting non trivially on $\mathcal{H}_m$ if and only if $\kappa_m=1$. Indeed, we can partition the whole space as
\begin{equation}\label{eq:B_subdivision}
    \BH = \bigoplus_{\kappa\in\{0,1\}^M} \mathcal{B}_\kappa 
\end{equation}
with $d_\kappa = \dim(\BHk) = \prod_m (d_m^2-1)^{\kappa_m}$. Clearly, if $\kappa = 0$, $\mathcal{B}_0 = \text{span}\, \mathds{1}$. In this work, such decomposition, is delineated as \emph{locality} from $\mathcal{H}$ to $\BH$. Specifically, we say that $A\in \BH$ is $\kappa$-local if $A \in \BHk$, and we associate to $A$ a locality vector $\ell_A \in \mathbb{R}^{2^M}$ defined element-wise by
\begin{equation}\label{eq:locality_vector}
    (\ell_A)_\kappa = \sum_{j=1}^{d_\kappa} \Trb{B_j A}^2
\end{equation}
for some Hermitian, orthonormal basis $\{B_j\}_{j=1}^{d_\kappa}$ of $\BHk$. Clearly, from the definition, a $\kappa$-local operator $A$ has locality vector $(\ell_A)_\lambda = \delta_{\kappa,\lambda} \|A\|^2_2$, where $\|\cdot\|_2$ is the Hilbert-Schmidt norm. We remark that similar quantities are not new in the context of PQCs, and in fact the vector $\ell_A$ is analogous to the purity measures defined in~\cite{Ragone24, Diaz23}.
Based on this, a derived notion of \emph{locality preservation} can be introduced for linear maps acting on $\BH$. In particular, given a map $\Lambda:\BH\to\BH$ and two subspaces $\BHk$ and $\mathcal{B}_\lambda$, we can measure the degree to which $\Lambda$ is able to put them in communication. The more subspaces are connected, the less locality preserving the map $\Lambda$ will be. This idea is captured formally in the following definition of a locality transfer matrix (LTM).
\begin{definition}[Locality transfer matrix]\label{def:LTM}
    Given a linear map $\Lambda:\BH\to\BH$, its locality transfer matrix $T$ is defined elementwise as
\begin{equation}\label{eq:LTM}
    T_{\kappa, \lambda} = \frac{1}{d_\kappa} \sum_{j=1}^{d_\kappa}(\ell_{\Lambda(B_j)})_\lambda
\end{equation}
for some Hermitian, orthonormal basis $\{B_j\}_{j=1}^{d_\kappa}$ of $\BHk$.
\end{definition}
In this formalism, \emph{locality preserving} transformations reflect all maps whose LTM coincides with the identity, i.e. $T_{\kappa, \lambda} = \delta_{\kappa, \lambda}$.
Trivially, unitary maps, separable with respect to the partition, namely $\mathcal{U}: \rho \mapsto U \rho U^\dagger$, where $U = \bigotimes_m U_m$, are locality preserving.

With this description in mind, we assume that $\mathcal{U}_\theta$ of \cref{eq:circuit_channel} is locality preserving, and hence describes an ideal operation limited to the local subsystems $\mathcal{H}_m$, while each channel $\mathcal{E}_l$ encodes both the operations which entangle the subsystems as well as any interaction between the system and the environment. This formally captures the idea of a variational quantum algorithm running on a real, noisy, device, where quantum computation can be realised very precisely within each subsystem, but is still inaccurate when dealing with more complex entangling operations. Furthermore, we assume that, within each subsystem and for all layers $l$, $\mathcal{U}_{\theta_l}$ is deep enough to form an approximate $2$-design over the local unitary groups $U(d_m)$. This assumption is justified by the fact that such local operations are relatively inexpensive, and that the necessary depth can be very small even for large systems, since in general it depends on $d_m$ rather than $d$.

A crucial property in the following analysis is the relation between the LTMs of the channel $\mathcal{E}_l$ and its Hermitian adjoint $\mathcal{E}_l^\dagger$ with respect to the Hilbert-Schmidt scalar product, namely $T$ and $T^\dagger$. Such relation can be characterized for generic linear maps $\Lambda:\BH\to \BH$ as a direct consequence of \cref{def:LTM}, and in particular we have $TD = (T^\dagger D)^t$, with $D_{\kappa, \lambda} = d_\kappa \delta_{\kappa, \lambda}$. For sake of readability, here we introduce a shorthand notation for the scalar product $(\cdot,\cdot)$ in $\mathbb{R}^{2^M}$ such that $T$ and $T^\dagger$ are \emph{also} Hermitian adjoint of one another, i.e.

\begin{equation}\label{eq:scalar_prod}
    (a,b) = a^tD^{-1}b = \sum_\kappa \frac{a_\kappa b_\kappa}{d_\kappa}.
\end{equation}
We refer the reader to \cref{app:loc_properties} for a more detailed discussion.
Having introduced the main tools, we now characterize the scaling of the variance $\Varr^L$ of $\Loss$ for quantum channels described by \cref{eq:circuit_channel}.  

\section{Results}\label{sec:results}

The overarching goal is to characterise the properties of the variance $\Varr^L$ as a function of the number $L$ of layers. To this end, we will provide a formal expression  for the general case; we will then use it to explicitly compute the variance of $\Loss$ in the deep circuit limit, as well as to set lower bounds for shallow circuits. We refer to \cref{fig:summary} for a schematic summary of the main results of this work.
\subsection{Loss variance calculation}

The first task is to derive a formal expression for the variance $\Varr^L$ in terms of the properties of $\rho, H$ and the intermediate channels, showing its relevance and main domains of applicability. Our first result, given in the following Proposition, serves as the foundation of the subsequent arguments.

\begin{proposition}[General formula]\label{prop:general_formula}
Let $\rho, H \in \BH$ and let $\Phi_\theta$ be a layered quantum channel as in described in \cref{sec:intro}. Then, we have
\begin{equation}\label{eq:general_formula}
    \mathbb{E}_{\theta} \left\{ \Trb{\Phi_\theta(\rho)H}^2 \right\} = \left(\ell_\rho, \prod_{l=1}^L T_l\,\ell_H\right)
\end{equation}
where $(\cdot,\cdot)$ is the scalar product defined in \cref{eq:scalar_prod}, and each $T_l$ is the LTM associated to the respective $\mathcal{E}_l^\dagger$.
\end{proposition}
Exploiting the fact that, under the assumptions of \cref{sec:intro}, each layer $\mathcal{U}_{\theta_l}$ forms a global unitary $1$-design, \cref{prop:general_formula} allows a direct calculation of $\Varr^L$ as a function of the LTMs of the intermediate channels as
\begin{equation}\label{eq:varr_0}
    \Varr^L = \left(\ell_\rho, \prod_{l=1}^L T_l\,\ell_H\right) - \frac{\Trb{H}^2}{d^2}.
\end{equation}
A proof of the aforementioned statements is given in \cref{app:preliminary_results}.

Note that this formulation provides an exact formula, which in principle gives access to the study of loss function concentration for any $L$. Indeed, this is the case when there are few subsystems, since in this case one can easily estimate and manipulate the matrices $T_l$ (see \cref{app:non_conv_var} for an example).
This approach becomes rather cumbersome in a general setting, for large systems, since the resources needed to represent the LTM can grow exponentially. Nevertheless, as shown in the following sections, \cref{prop:general_formula} can still be profitably used as a theoretical tool, as it allows characterizing $\Varr^L$ in both the deep and shallow circuit limit. 

\begin{figure*}
    \includegraphics[width=0.9\linewidth]{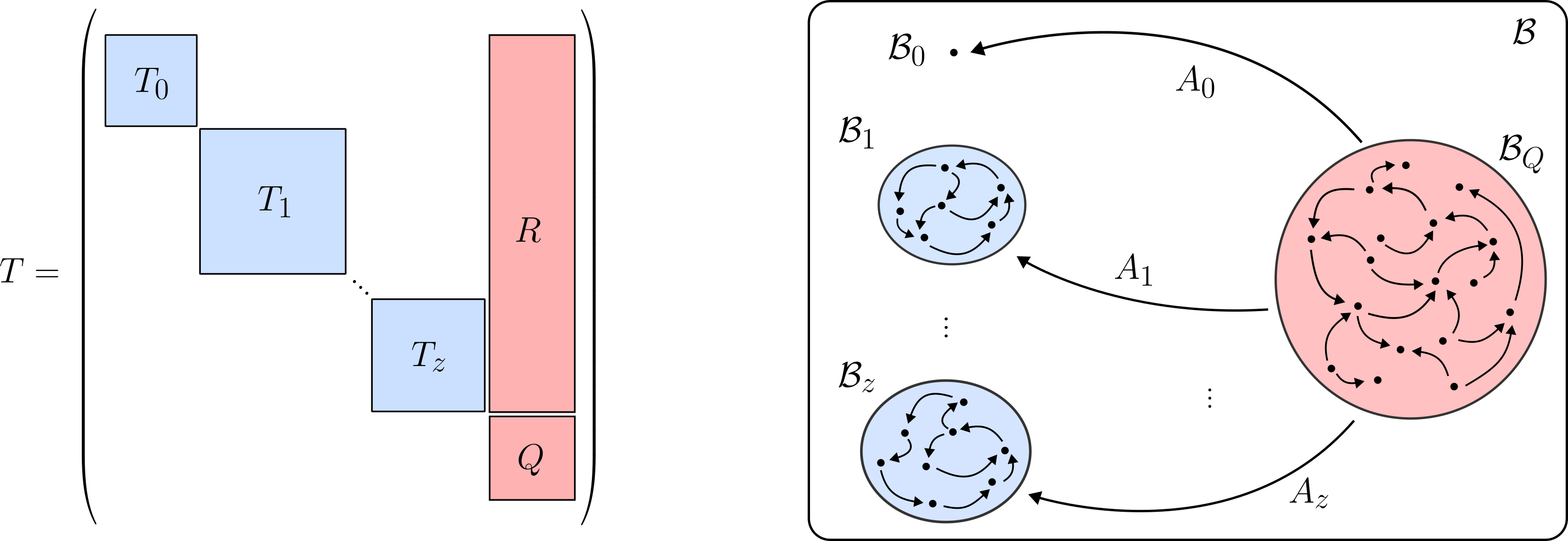}
    \caption{Graphical representations of the stochastic process describing $\Varr^L$. On the left, we show the structure of the general locality transfer matrix (LTM), highlighting the decomposition into irreducible components \cite{Seneta06}. Light blue blocks represent irreducible, essential components of $T$, while red blocks are related to inessential ones. In particular, $Q$ represents the collection of all irreducible, inessential components of $T$ and $R$ their relation with the essential components $T_z$. On the right, the same process is represented graphically, in terms of the local subspaces $\BHk$. Here, each dot represents a single subspace, while the arrows represent the adjoint action of the channel $\mathcal{E}$. Essential and inessential components share the same colour code of $T$.}
    \label{fig:T_and_B}
\end{figure*}

\subsection{Deep circuit limit}
As the circuits become deeper, it is natural to expect the contribution of the leading eigenvectors of $T_l$ to become dominant, as repeatedly multiplying them gives rise to a process analogous to power iteration methods. The structure of such eigenvectors captures several interesting properties of the interaction between the unitary and non-unitary parts of $\Phi_\theta$, particularly \emph{absorption}, which can only arise in this picture. 

For simplicity, let us focus on \emph{homogeneous} channels, i.e. we fix $\mathcal{E}_l = \mathcal{E} \; \forall l$, and consider, without loss of generality, a traceless observable $H$, $\Trb{H}=0$. In this case, \cref{eq:varr_0} becomes $\Varr^L = \left(\ell_\rho, T^L\ell_H\right)$. Note that, by construction, $T$ is a non-negative matrix, namely $T_{\kappa,\lambda} \geq 0\;\forall {\kappa,\lambda}$, and thus it can always be expressed in canonical, block-upper triangular form, where each of the diagonal blocks is irreducible \cite{Seneta06}. Throughout this work, irreducible components of $T$ will be regarded as \emph{essential} if they cannot lead outside the block, and will be denoted by $T_z$. 
Otherwise, an irreducible block will be deemed \emph{inessential}, and their collection will be denoted by $Q$. A useful pictorial representation of these possibilities is shown in \cref{fig:T_and_B}, together with the matrix canonical form.  A foundational reference for this definition can be found in Ref.~\cite{Seneta06}, while a more detailed discussion is proposed in \cref{appendix:proof_main}. 

Given this structure, it will be particularly useful to denote by $\BHz = \bigoplus_{\kappa \in T_z} \BHk$ the union of subspaces put in communication within $T_z$, by $d_z$ their total dimension, and given $A\in \BH$, by $(\ell_A)_z = \sum_{\kappa \in T_z} (\ell_A)_\kappa$ the corresponding locality.
Regardless of the specific channel $\mathcal{E}$ used, some general properties of the blocks $T_z$ can be identified. For instance, due to trace preservation, the trivial subspace $\mathcal{B}_0$ always forms an essential component of $T$, which we denote by $T_0$. Moreover, since such blocks are essential, it follows that $\mathcal{E}^\dagger(\BHz) \subset \BHz$, and complete positivity ensures that all $T_z$ must be contractive in the sense of the spectral radius, i.e. $\rho(T_z) \leq 1$. Moreover, when the equality holds, the left eigenvector $v_z$ of $T_z$ of the dominant eigenvalue can be explicitly computed, yielding $(v_z)_\kappa = 1\,\forall\kappa$. In contrast, its contribution vanishes under repeated applications of the channel when the condition is not satisfied. This suggests that, as $L\to\infty$, the general form for the variance reads

\begin{equation}\label{eq:v_infty_general}
    \Varr^\infty = \sum_{z \,|\, \rho(T_z) = 1} (\ell_\rho,w_z) (\ell_H)_z + (\ell_\rho,w_z) (A \ell_H)_z,
\end{equation}
where $w_z$ denotes the right eigenvector of the leading eigenvalue of $T_z$ and $A$ is a matrix of the same shape as $R$ in \cref{fig:T_and_B}. 
Such reasoning is formally encapsulated in the following Theorem.

\begin{theorem} [Deep circuit variance]
\label{thm:main}
Let $\rho, H \in \BH$ and let $\Phi_\theta$ a be layered quantum channel as in described in \cref{sec:intro}. Then the Cesàro average of $\Varr^L$ converges to \cref{eq:v_infty_general}, and we have
\begin{equation}
\label{eq:main:cesaro_limit}
    \left|\frac{1}{L} \sum_{l=0}^L \Varr^l -\Varr^\infty \right| \in O\left(e^{-\beta L} \|H\|^2_2\right),
\end{equation}
for some constant $\beta > 0$. Additionally, if all essential blocks are aperiodic (i.e. with period $p=1$), then $\Varr^L$ is convergent, and we have
\begin{equation}
\label{eq:main:simple_limit}
    \left|\Varr^L -\Varr^\infty \right| \in O\left(e^{-\beta L} \|H\|^2_2\right),
\end{equation}
where the absorption matrix is given by $A = R(\mathds{1}-Q)^{-1}$. 
\end{theorem}

As a special case, if the intermediate channel can be reduced unitarily to a tensor product of single qubit channels, then also $w_z$ of \cref{eq:v_infty_general} can be explicitly computed, yielding the following Corollary.

\begin{corollary}\label{cor:single_qubit_noise_main}
If the intermediate channel takes the form $\mathcal{E} = \mathcal{N}\circ\mathcal{W}$, where the noise channel $\mathcal{N} = \bigotimes_m \mathcal{N}_m$ is the composition of single qubit channels and $\mathcal{W}$ is unitary, then 
\begin{equation}\label{eq:V_infty}
    \Varr^\infty = \sum_{z \,|\, \rho(T_z) = 1} \frac{(\ell_\rho)_z (\ell_H)_z}{d_z} + \frac{(\ell_\rho)_z (A \ell_H)_z}{d_z}.
\end{equation}
\end{corollary}

The proof of \cref{thm:main} and \cref{cor:single_qubit_noise_main} can be found in \cref{appendix:proof_main}. 

It is worth noting that, the loss variance of the circuit need not converge. In fact, while it is bounded (i.e. $\Varr^L \leq \|H\|_\infty^2 \; \forall L$), different circuit sub-sequences can in general lead to different limits. This is connected to the presence of \emph{cycles} of period $p>1$ in each $T_z$, and can arise, for instance, when the entangling operation is not chosen carefully with respect to the partition of $\mathcal{H}$. In those cases, the value of $\Varr^L$ will depend on which equivalence class of integers modulo $p$ the depth $L$ belongs to. 
This is put into an effective example and further discussed \cref{app:non_conv_var}.
For simplicity, in the following, we will assume \emph{aperiodicity} for each irreducible block.

We start the analysis of \cref{thm:main} by unpacking \cref{eq:V_infty}. In this equation, $\Varr^\infty$ is shown to depend on four important quantities, namely the invariant subspaces $\BHz$ of $\Phi_\theta^\dagger$, the respective locality of $\rho$ and $H$, together with the matrix $A$. As a reminder, in order for this structure to arise, the invariant subspaces of $\mathcal{U}_\theta^\dagger$ and $\mathcal{E}^\dagger$ have to be \emph{well-aligned}, so that non-trivial subspaces $\BHz \subset \BH$ are present. Failure to achieve this will cause the emergence of noise-induced concentration in the sense of Ref. \cite{Wang21}, as shown in \cref{sec:methods:recover_limits}. This extends the notion of alignment already introduced in the literature \cite{Larocca24} for $\rho, H$ and unitary circuits. Moreover, we observe that such subspaces act as \emph{attractors} for the variance, as each summand not only depends on the local components of $\rho$ and $H$, but also on the components of $H$ belonging to $\mathcal{E}(\BHz) \cap \mathcal{B}_Q$.
In this sense, $A$ can be interpreted as an absorption matrix, which quantifies the extent of the contribution of such terms. We remark that this contribution is always non-negative, and it is intimately related to the contractivity properties of $\mathcal{E}^\dagger$. Overall, this suggests that appropriate non-unitary layers will alleviate the concentration typical of unitary circuits by a mechanism that allows to bring the contribution of the components of $H$ belonging to strictly contractive, high dimensional subspaces, to non-contractive, smaller dimensional ones.

As shown in \cref{sec:methods:recover_limits}, \cref{eq:V_infty} can be used to recover deep circuit limit variance in the unitary and strictly contractive setting, consistent with existing literature \cite{Ragone24, Fontana24, Mele24, Singkanipa24}. In particular, regarding strictly contractive, non-unital noise, we show how previous results can be expressed as the contribution to the absorption term of $\mathcal{B}_0$, which always forms a norm-preserving subspace due to trace preservation of $\Phi_\theta$. However, it is crucial to remark that this specific contribution is not due to the retention of any computational power to the PQC, since the dependence on the initial state is completely lost, but rather to the competing effects between the drive of $\mathcal{U}_\theta$ and $\mathcal{E}$ towards the respective, different fixed points. An example of this phenomenon is provided in \cref{sec:methods:noise_ent}, where we study how extensive entanglement can further exacerbate the issue. This is in stark contrast with the general case discussed above, i.e. the absorption terms of $\BHz$, $z>0$, as, although more difficult to realize, they keep a non-trivial dependence on the initial state, and as such it can be said to genuinely avoid concentration if $d_z$ scales appropriately.

The results showed so far were based on the analysis of the dominant eigenvectors of $T$. When $L$ is not deep enough to reach convergence to its asymptotic limit, we need to characterize better the behaviour of $\Varr^L$. This is the subject of the following Section.

\subsection{Lower bounds on ``shallow" circuits}

While for shallow circuits we cannot rely on the spectral properties of $T$ to determine $\Varr^L$, we can still use knowledge of the convergence speed to $\Varr^{\infty}$ to set general lower bounds. Intuitively, such bounds can be obtained by preventing the variance to reach its stationary state, which can be done only if the exponential upper bounds appearing in \cref{thm:main} are sufficiently loose, namely $\beta L \in O(\log n)$. This implies either that the circuit is \emph{shallow}, i.e. there are not enough layers to reach the asymptotic value $\Varr^{\infty}$, or \emph{effectively shallow}, i.e. the mixing speed $\beta$ of $T$ is slowed down according to $\beta \in O(\log n/L)$ so that $\Varr^{\infty}$ is never reached, regardless of $L$.
As this speed is related to the amount of correlations with respect to the partition introduced by the channel, we study this scenario in the limit of $\mathcal{E}^\dagger$ being close to locality preserving. This idea is formally captured in the following Theorem, which extends its application to generally non-homogeneous channels.

\begin{theorem}[General lower bound]\label{thm:lower_bound}
    Let $\rho, H \in \BH$ and consider a sequence of quantum channels $\{\mathcal{E}_l\}_{l=1}^L$, and let $\{T_l\}_{l=1}^L$ be the respective LTMs. Finally let $K\subset \{0,1\}^M$ denote a subset of indices, and by $\alpha_l = \min_{\kappa\in K} (T_l)_{\kappa,\kappa}$. Then
    \begin{equation}\label{eq:lower_bound_main}
        \Varr^L \geq  \alpha^L (\ell_\rho,\ell_{\mathcal{K}(H)}),
    \end{equation}
    where $\mathcal{K}:\BH\to\mathcal{B}_K$ is the projector onto $\mathcal{B}_K = \bigoplus_{\kappa \in K} \BHk$ and $\alpha$ is the geometric mean of $\alpha_l$. 
\end{theorem}

Depending on the scaling of $\alpha$ and the dimensions $d_m$ of the subsystems, \cref{eq:lower_bound_main} can provide a meaningful lower bound. For instance, focusing on the case of subsystems with constant dimension, we have the following Corollary.

\begin{corollary}[Lower bound examples]\label{cor:small_angle}
    Let $\mathcal{H} = \bigotimes_{m=1}^M \mathcal{H}_m$, $d_m \in \Theta(1)$. If either of the conditions 
    \begin{enumerate}[(a)]
        \item $\alpha > 0$, $\alpha \in \Omega(1)$ and $L\in O(\log n)$,
        \item $\alpha = 1-f(n,L)$, $f \in O(\log n/L)$ and {$L \in \Omega(\log^{1+\epsilon}n)$} for some arbitrary $\epsilon > 0$
    \end{enumerate}
    is satisfied, then
    \begin{equation}
    \label{eq:small_angle:bound}
        \Varr^L \geq F(n) (\ell_\rho, \ell_{{\mathcal{K}(H)}}),
    \end{equation}
    where $F(n) \in \Omega(1/\text{poly}(n))$.  
\end{corollary}

The conditions of \cref{cor:small_angle} reflect the aforementioned scenarios; in particular, condition $(a)$ ensures the absence of concentration for \emph{shallow} circuits, both unitary and noisy. Specifically, this holds true whenever $ 0 < \alpha \in \Omega(1)$, indicating that the intermediate channel does not become increasingly rapidly entangling, as the problem size grows. As a notable example, this condition is satisfied for brickwork circuits, equipped with local noise and a local observable \cite{Mele24}. Similarly, condition $(b)$ reflects the absence of concentration for \emph{effectively shallow} quantum channels. A significant example is that of finite local-depth circuits (FLDCs) \cite{Zhang24}. We can interpret condition $(b)$ as a limit to the mixing speed of $T$ by noting that, in the homogeneous case, it is equivalent to the more explicit relation $|1-\lambda| \in O(\log n/L)$ for all the eigenvalues $\lambda$ of $T$ by Gershgorin circles theorem \cite{Meyer00}, which directly implies $\beta \in O(\log n/L)$.

Interestingly, since these results have been obtained by imposing $T \approx \mathds{1}$, they have a strong resemblance with small angle initialization strategies \cite{Wang23, Zhang22}, which similarly hinge on identity manipulation. In fact, while the primary concern of \cref{thm:lower_bound} is noise, it could still be regarded as a theoretical foundation of such initialization strategies. In the following Section, we show a deep relation between these types of smart initializations for noiseless circuits and the properties of noisy layers.

\subsection{Connection with small angle initializations}\label{subsect:small_angles}

Here we expand on the concept of small angle initialization introduced in Refs. \cite{Zhang22, Wang23}.
In particular, we establish a general relationship between the insights gained from controlling loss concentration in noisy circuits (as presented in \cref{thm:lower_bound}) and BP mitigation strategies that typically only apply to ideal circuits.

The main idea behind small angle initialization strategies considers a layered quantum circuit $U_\theta = \prod_l U_{\theta_l}$ where absence of concentration for a given initialization distribution is shown. This is typically very peaked around $0$, with variance scaling inversely to the number of layers: $\sigma^2 \in O(1/L)$. The core idea of all such strategies relies on \emph{identity manipulation}, i.e. on choosing initialization distributions such that $U_{\theta} \approx \mathds{1}$ with high probability. This introduces a sizable variance to the circuit, at the price of having a large bias towards the identity in the quantum model. 

A similar structure can be defined in our framework by considering quantum channels $\Phi_{\theta, \phi}$, as in \cref{eq:circuit_channel}, where the intermediate channels $\mathcal{E}_{\phi_{l}}$ are now parameterized. This differs from the main idea of small angle initialization, as the local components $\mathcal{U}_{\theta_l}$ of the channel remain Haar random, and instead it is the allowed channels $\mathcal{E}_\phi$ that get restricted. Intuitively, this will lead to a different model bias for small angles, i.e. $\Phi_{\theta, \phi} \approx \mathcal{U}_\theta$.
We name this model Quantum Residual Network (QResNet), as we interpret the large identity component of $\mathcal{E}_\phi$ as a \emph{skip-connection}, in analogy to classical Residual Networks \cite{He15}. 
Indeed, this structure is enough to avoid concentration when a small angle initialization strategy is used. To see this, we remark that while a constant channel was used to derive \cref{prop:general_formula}, it can be readily generalized to parameterized channels $\mathcal{E}_\phi$, as long as the parameters $\phi_l$ are independent. In that case it is sufficient to use $\mathbb{E}_\phi\{ T_\phi \}$ in place of $T$, where $T_\phi$ is the LTM of $\mathcal{E}^\dagger_\phi$ (see \cref{appendix:lower_bounds} for a proof). Exploiting this, we can derive the following Proposition.

\begin{proposition}[QResNet]\label{prop:qresnet}
    Let $\mathcal{E}_\phi(\cdot) = e^{i\phi G} \cdot e^{-i\phi G}$ be a unitary entangling gate, and let $\mu, \sigma^2$ be the mean and variance of the initialization distribution of $\phi$. Then, if $\mu=0$, ${\sigma^2 \in O(\log n/\|G\|_2^2L)}$, and $L\in\Omega(\log^{1+\epsilon} n)$ we have
    \begin{equation}
        \Varr^L \geq F(n) (\ell_\rho, \ell_H),
    \end{equation}
    where $F(n) \in \Omega(1/\text{poly}(n))$.
\end{proposition}

This result represents an application of \cref{cor:small_angle} $(b)$ in the case of unitary, parameterized intermediate channels.
Note that \cref{thm:lower_bound}, which is the backbone of \cref{prop:qresnet}, is not limited to unitary circuits, but is applicable to generic quantum channels. Indeed, if we take $\mathcal{E}$ to be a noise model, \cref{cor:small_angle} $(b)$ may be analogously interpreted as a condition on the noise rates to avoid concentration. This showcases a connection between QResNets that can avoid BP and noise models not strong enough to cause NIBP. Formally, this is captured by the following Proposition. 

\begin{proposition}[Noise map and QResNet]\label{prop:noise_small_angle_equiv}
    Let $\{E_\phi\}_\phi$ be an ensemble such that $\mathcal{E}(\rho) = \mathbb{E}_\phi \{ E_\phi \rho E_\phi^\dagger\}$ is a quantum channel. Further, denote by $T_\phi$ the transfer matrix associated to each $\mathcal{E}_\phi^\dagger(\cdot) = E^\dagger_\phi \cdot E_\phi$, and by $T$ the transfer matrix of $\mathcal{E}^\dagger$. Then we have
    \begin{equation}
\mathbb{E}_\phi \{T_\phi\} \geq T,
    \end{equation}
elementwise, with equality if and only if $\mathcal{E}$ is unitary.
\end{proposition}

This result can be interpreted as follows: if a noise map satisfies \cref{cor:small_angle}, then there exists a QResNet associated to it that is able to avoid BP. 
As an example, assume that the channel $\mathcal{E} = e^{\Delta t \mathcal{L}}$ is obtained as a solution of the Lindblad equation at time $\Delta t \ll 1$, where $\mathcal{L}(\rho) = \sum_i L_i \rho L_i - 1/2\{L_i^2,\rho\}$, $L_i=L_i^\dagger \, \forall i$, i.e. we consider weak Lindbladian noise \cite{DiBartolomeo23}. Then, there always exist a unitary, stochastic unravelling $\{U_\phi\}_\phi$, such that $\mathcal{E}(\rho) = \mathbb{E}_\phi \{ U_\phi \rho U_\phi^\dagger\}$ \cite{Adler07}. From \cref{prop:noise_small_angle_equiv}, the variance of the model obtained from the ensemble by promoting $\phi$ to variational parameters is lower-bounded by the variance of the channel, and hence it provides a BP free QResNet for weak enough noise.
Interestingly, \cref{prop:noise_small_angle_equiv} can be equally applied if the ensemble is not unitary, which can happen if the noise channel $\mathcal{E}$ is non-unital. This effectively extends the framework of small angle initialization to non-unitary quantum models, e.g. those based on linear combination of unitaries (LCU) or analogous techniques \cite{Heredge24}.

Finally, it is important to note that stochastic unravellings are not unique, allowing multiple QResNets (unitary or non-unitary) to be derived from the same channel. Additionally, different choices among such models may lead to distinct variances \cite{Weisman93}, all bounded from below by the variance of $\mathcal{E}$. Indeed, any stochastic unravelling is connected to the physical channel solely through their first moment, while higher-order moments are relevant only in the context of the associated QResNet, as they do not represent otherwise measurable quantities.

\section{Discussion}

The study of loss function concentration is a central topic in variational quantum computing. While the description of this effect in the absence of noise has been recently formulated using Lie-algebraic theory \cite{Ragone24, Fontana24, Diaz23}, this approach inevitably fails in the general setting of non-unitary circuits, where the group structure description is lost. 
In this work, we employ non-negative matrix theory to derive a general expression for $\Varr^L$ in quantum circuits composed of local 2-designs interleaved with arbitrary quantum channels. 
This circuit architecture is chosen for its combination of analytical accessibility and physical relevance: it is rich enough to capture qualitatively novel variance behavior, while structured enough to address fundamental questions about the interplay between noise and unitary layers in the emergence of loss concentration.

In particular, within this framework, the variance is formulated in terms of the channel’s locality transfer matrix (LTM), which in turn enables the precise computation of $\Varr^\infty$ in the deep-circuit limit.
Specifically, the observed structure of $\Varr^\infty$ brings out a new mechanism, which we call \emph{absorption}, whereby components of $H$ pertaining to strictly contractive subspaces of $\mathcal{B}_Q$ can augment the variance of the model by coupling with non-contractive ones through the non-reversible action of the quantum noise channel. Such result is supported by numerical simulations, provided in \cref{sec:methods:noise_ent} and \cref{sec:unital_noise_examples}.

This indicates that a mixed configuration of ideal and noisy qubits, as well as a partial realization of error correction, could potentially outperform purely ideal or purely noisy systems in terms of variance. This is especially important in the early stages of fault-tolerant quantum computing, where a large quantity of noisy qubits can be utilized, but the availability of ideal qubits is still constrained \cite{Preskill18, Katabarwa24}.
This may also be beneficial in characterizing variational quantum algorithms executed on hardware, where different qubits may experience varying error rates \cite{DiBartolomeo23}. In this context, $\Varr^\infty$ can be used to approximate $\Varr^L$ in regimes where $L$ is sufficiently large to significantly impact some qubits but not yet others. 

Furthermore, this approach clarifies the role of unitality in noise-induced barren plateaus, shifting the focus on the more relevant aspect of \emph{contractivity} and \emph{alignment} between invariant subspaces of the unitary layers $\mathcal{U_\theta}$ and strictly contractive subspaces of the noise map $\mathcal{E}$.

Additionally, we introduced a general lower bound on the variance of noisy circuits. This bound is derived by restricting the mixing speed of $\Varr^L$ to prevent it from reaching its asymptotic limit.
We subsequently employ this approach to introduce QResNets, an initialization strategy analogous to small angle initializations \cite{Zhang22, Wang23, Park24, Park24_2} that can effectively mitigate the occurrence of barren plateaus. Moreover, we demonstrate that an analogous procedure can be applied to noise maps, establishing a formal connection between weak noise and QResNets. This enables us to derive BP-free QResNets as stochastic unravellings of sufficiently weak noise maps. Notably, since these unravellings are not necessarily unitary, this approach can yield BP-free architectures beyond unitary circuits, encompassing more complex models such as those employing linear combinations of unitaries (LCU) \cite{Heredge24}.

Future research may extend our findings by relaxing the local 2-design property of the unitary layer. Such a generalization would broaden the domain of applicability of our results beyond hardware-specific designs, but it would primarily offer a technical refinement without altering the conceptual insights gained in this work.
Furthermore, recent studies have established a connection between the absence of concentration and classical simulability for both ideal \cite{Cerezo24} and noisy \cite{Mele24} quantum circuits. While these results often focus on strictly contractive noise models, our work suggests a potential avenue for combining these concepts, extending their validity to more complex noise environments.
A representation of our contributions is depicted in \cref{fig:summary}: the non negative matrix approach offers a new and timely answers to several open questions related to concentration phenomena in quantum circuits and provides valuable insights into the optimal utilization of near-term and early fault-tolerant quantum devices, thus guiding the community towards the effective application of the variational quantum computation framework.

\begin{acknowledgments}
G.C. thanks F. Benatti, G. Nichele, M. Vischi, G. Di Bartolomeo, A. Candido, S. Y. Chang, M. Sahebi and M. Robbiati for useful discussions, and the CERN Openlab agreement with the University of Trieste. G. C. acknowledges financial support from University of Trieste, INFN and EU Erasmus+ Traineeship Programme.
M.G. is supported by CERN through
the CERN Quantum Technology Initiative. A.B. acknowledges support from the University of Trieste, INFN, the PNRR MUR project PE0000023 - NQSTI and the EU EIC Pathfinder project QuCoM (GA 101046973).
\end{acknowledgments}


\appendix

\section{Limiting cases of the deep circuit formula}\label{sec:methods:recover_limits}

Starting from the analysis of \cref{sec:results}, we show how we can use \cref{eq:V_infty} to recover previously known results by considering specific classes of quantum channels $\mathcal{E}$. First we point out that, being an absorption term, the last term in \cref{eq:V_infty} vanishes for unitary dynamics, which is reversible by definition. Formally, we have the following Corollary of \cref{thm:main}.

\begin{corollary}[Deep, unitary circuits]
\label{cor:deep_unitary}
Let $\rho, H \in \BH$ and let $\Phi_\theta$ a be layered quantum channel as described in \cref{sec:intro}, where $\mathcal{E}(\cdot) = W\cdot W^\dagger$, $W\in U(d)$ is an arbitrary unitary transformation. Then we have
\begin{equation}
\label{eq:BP_unitary}
    \Varr^\infty = \sum_{z>0} \frac{(\ell_{\rho})_z (\ell_{H})_z}{d_z},
\end{equation}
where $\BHz$ are invariant subspaces of $\mathcal{E}$ which can be expressed as the direct sum of $\BHk$.
\end{corollary}
This Corollary contains many interesting properties of the variance in the deep circuit limit. First, it captures the necessity of \emph{alignment} between $\rho$, $H$ and $\Phi_\theta$ in order to achieve a substantial variance, i.e. $\rho$ and $H$ need both to have non-negligible components on the same invariant subspace $\BHz$. Due to the structure of the channel, this idea is extended to the components $\mathcal{U}_\theta$ and $\mathcal{E}$, whose invariant subspaces need to align in order to keep the dimension $d_z$ of $\BHz$ from being exponentially large. Indeed, while $\BH_0$ is always an invariant subspace satisfying \cref{cor:deep_unitary}, misalignment of the local and entangling parts of the circuit could result in the entirety of the remaining space falling under a single, irreducible component of dimension $d^2-1$. In such cases we get
\begin{equation}
\label{eq:BP_unitary_irreducible}
    \Varr^\infty = \frac{\left(\|\rho\|_2^2 - 1/d\right) \left(\|H\|_2^2 -\Trb{H}^2/d\right)}{d^2-1},
\end{equation}
which implies the presence of BP regardless of $\rho$ and $H$ as long as $\|H\|_2 < O(d)$. Indeed, one can interpret the misalignment of $\mathcal{U}_\theta$ and $\mathcal{E}$ as introducing an excess of expressibility, which is known to lead to exponential concentration \cite{Holmes22}.

As a complementary remark, we point out that, conversely to the above, the first term in \cref{eq:V_infty} pertains to non-contractive subspaces, and as such, vanishes if $\mathcal{E}^\dagger$ is strictly-contractive in, at least, one direction in each $\BHz$.
This is formalized in the following Corollary of \cref{thm:main}.

\begin{corollary}[Deep, noisy circuits]\label{cor:deep_noisy} Let $\rho, H \in \BH$ and let $\Phi_\theta$ a be layered quantum channel as described in \cref{sec:intro}, where $\mathcal{E}$ is such that $\|\mathcal{E}(A)\|_2 < \|A\|_2$, for at least one $A\in\BHk \subset \BHz \, \forall z > 0$ . Then, we have
\begin{equation}\label{eq:BP_noisy_nonunital}
    \Varr^\infty = \frac{(A \ell_H)_0}{d}.
\end{equation}
In particular, if the channel is unital, $\Varr^\infty=0$.
\end{corollary}
Note that, even if \cref{eq:BP_noisy_nonunital} is inversely proportional to $d=2^n$, $\Varr^\infty$ is not necessarily exponentially suppressed, as in general the contribution of the observable increases with the same speed, i.e. $\|H\|_2^2 \sim d$.
As before, this Corollary captures the main features of noise-induced barren plateaus (NIBP). In fact, it is clear that strictly contractive channels with a unique fixed point, fall into the assumptions of \cref{cor:deep_noisy}, and therefore exhibit some form of concentration. However, \cref{cor:deep_noisy} is not limited to them, extending the noise-induced concentration to a wider class of noise maps, which crucially depend on the structure of the unitary part of the channel $\mathcal{U}_\theta$. According to the method introduced here, this can clearly be interpreted as a consequence of the interaction between noise and unitary layers, which effectively ``spread" the contractive effect of $\mathcal{E}$ to the whole irreducible component.  
Unital channels will suffer most severely from NIBP, since in that case the absorption term in \cref{eq:BP_noisy_nonunital} vanishes. Contrarily, as pointed out in the literature \cite{Mele24, Singkanipa24}, non-unital channels may avoid the exponential concentration. Here, we have shown how these results obtained in the literature can be seen as the contribution to the absorption term of $\mathcal{B}_0$ which cannot be strictly contractive due to trace preservation of $\Phi_\theta$.

\section{Locality and locality transfer matrix properties}\label{app:loc_properties}

In order to show the main properties of locality vectors and locality transfer matrices (LTM), it is convenient, for each subsystem $\mathcal{H}_m$, to fix an orthonormal basis $\{P^{(m)}_j\}_{j=1}^{d_m-1}$ of $\mathcal{B}_m$ composed of traceless, Hermitian operators together with the identity operator, each normalized with respect to the Hilbert-Schmidt norm, namely $P^{(m)}_0 = \mathds{1}/\sqrt{d_m}$ , $P^{(m)}_{j_m} = P^{(m)\dagger}_{j_m} \, \forall {j_m}$ ,$\Trb{P^{(m)}_{j_m} P^{(m)}_{k_m}} = \delta_{{j_m}{k_m}}$.
Starting from these, we can build an orthonormal basis $\{P_j\}_j$ for the whole space by means of tensor products. Each basis element will be labelled by the multi-index $j = (j_1, ... , j_M)$, where the entries $j_m$ refer to an element of the local bases, and hence $j_m \in \{0,..., d_m^2-1\}$.  Such a basis will be dubbed a \emph{local} basis.  As an example, if each $\mathcal{H}_m=\mathbb{C}^2$ is a qubit, the normalized Pauli strings form a local basis for $\BH$. Given a local basis $\{P_j\}_j$, it is possible to group the elements in disjoint sets. In particular, for any given binary string $\kappa = \{0,1\}^M$, we can collect in the set $S_\kappa$ all basis elements acting non trivially on $\mathcal{H}_m$ if and only if $\kappa_m=1$.  For practical reasons, we introduce the indicator function $\deltat{i}{\kappa}$ for the set $S_\kappa$, defined by
\begin{equation}
    \deltat{i}{\kappa} = \prod_{m=0}^M \Tilde{\delta}_{i_m,\kappa_m}=\begin{cases} 
      1 & \text{if} \;\; \kappa_m = 0 \Leftrightarrow i_m =0 \; \forall m \\
      0 & \text{otherwise}
   \end{cases}
\end{equation}
From the definition, we can derive some simple properties of this function.
\begin{lemma}\label{lemma:indicator_func}
    The indicator function $\deltat{i}{b}$ has the following properties:
    \begin{equation}
        \sum_i \deltat{i}{\kappa} = d_{\kappa}, \;\;\;\; \sum_{\kappa} \deltat{i}{\kappa} = 1, \;\;\;\; \sum_{\kappa \in K} \deltat{i}{\kappa}\deltat{i}{\lambda} = \deltat{i}{\lambda} \sum_{\kappa \in K} \delta_{\kappa,\lambda} 
    \end{equation}
    where $d_\kappa = \prod_{m=0}^M (d_m^2-1)^{\kappa_m}$, $K \subset \{0,1\}^M$, and $\delta_{\kappa,\lambda}$ is the usual Kronecker delta.
\end{lemma}
\begin{proof}
    All the results follow directly from the definition.
\end{proof}
Using this notation, we can express the \emph{locality} $\ell_A$  of some operator $A\in\BH$ and \emph{locality transfer matrix} $T$ of a linear map $\Lambda:\BH\to\BH$ defined in the main text as 

\begin{equation}
    (\ell_A)_{\kappa} = \sum_{j} \Trb{P_j A}^2 \deltat{j}{\kappa}
\end{equation}
and
\begin{equation}\label{eq:alternate_LTM}
     T_{\kappa, \lambda} = \frac{1}{d_\lambda} \sum_{i, j} \Trb{P_i\Lambda(P_j)}^2 \deltat{i}{\kappa} \deltat{j}{\lambda}
\end{equation}

respectively. It is immediately realized that both these quantities are basis-independent.

\begin{lemma}\label{lemma:basis_independence_T}
Given a bounded operator $A\in\BH$ and a partition into subsystems, the locality vector $\ell_A$ is uniquely defined, i.e. it does not depend on the choice of local basis. Similarly, given a quantum channel $\mathcal{E}:\BH\to\BH$. the corresponding locality transfer matrix $T$ is uniquely defined.
\end{lemma}

\begin{proof}
    Let $\{P_{j}\}$ and $\{B_{i}\}$ be two local bases for a given subsystem. Then
    \begin{equation}
    \begin{split}
        &\sum_{j} \Trb{P_{j}A}^2 \deltat{j}{\kappa}=\sum_{j} \Trb{\sum_{i} \Trb{B_i P_j} B_{i}A}^2 \deltat{j}{\kappa}\\
        = &\sum_{j} \sum_{i, i'} \Trb{B_iP_j} \Trb{B_{i'} P_j} \Trb{B_{i}A} \Trb{B_{i'}A} \deltat{j}{\kappa}\\
        = &\sum_{i, i'} \left(\sum_{j}\Trb{B_iP_j} \Trb{B_{i'} P_j} \deltat{j}{\kappa} \right)\Trb{B_{i}A} \Trb{B_{i'}A}\\
        = &\sum_{i} \Trb{B_{i}A}^2 \deltat{i}{\kappa}
    \end{split}
    \end{equation}
    where the last equality is due to
    \begin{equation}
    \begin{split}
    &\sum_{j}\Trb{B_iP_j} \Trb{B_{i'} P_j} \deltat{j}{\kappa}=\\ &=  \prod_{m=1}^M \left( \sum_{j_m=1}^{d_m} \Trb{B_{i_m'} P_{j_m}}\Trb{B_{i_m} P_{j_m}}\Tilde{\delta}_{j_m,\kappa_m}\right)\\ & = \prod_{m=1}^M \left(  \Trb{B_{i_m'} B_{i_m}}\Tilde{\delta}_{i_m,\kappa_m}\right) = \delta_{i,i'} \deltat{i}{\kappa} 
    \end{split}
    \end{equation}
    which holds since both $\{P_{j_m}\}$ and $\{B_{i_m}\}$ are orthonormal bases of $\mathcal{B}_m$ by definition of local basis. A totally analogous calculation yields the same result for the locality transfer matrix $T$.
\end{proof}
Thanks to the formulation of \cref{eq:alternate_LTM}, the relation between the LTM of a map $\Lambda$ and the Hermitian adjoint $\Lambda^\dagger$ with respect to the Hilbert-Schmidt scalar product can be seen. In particular, we have
\begin{equation}\label{eq:detailed_balance}
\begin{split}
    d_\lambda T_{\kappa, \lambda} &= \sum_{i, j} \Trb{P_i\Lambda(P_j)}^2 \deltat{i}{\kappa} \deltat{j}{\lambda}\\
    &= \sum_{i, j} \Trb{\Lambda^\dagger (P_i) P_j}^2 \deltat{i}{\kappa} \deltat{j}{\lambda} = d_\kappa T^\dagger_{\lambda,\kappa},
\end{split}
\end{equation}
which can be compactly written in matrix form as $TD = (T^\dagger D)^t$, with $D_{\kappa, \lambda} = d_\kappa \delta_{\kappa, \lambda}$. For sake of readability, here we introduce a shorthand notation for the scalar product $(\cdot,\cdot)$ in $\mathbb{R}^{2^M}$ such that $T$ and $T^\dagger$ are \emph{also} Hermitian adjoint of one another, i.e.
\begin{equation}
    (a,b) = a^tD^{-1}b = \sum_\kappa \frac{a_\kappa b_\kappa}{d_\kappa}.
\end{equation}
This trivially follows from the chain $(a,Tb) = a^t D^{-1}Tb=a^tD^{-1}D(T^\dagger)^t D^{-1}b = (T^\dagger a)^tD^{-1}b = (T^\dagger a, b)$.

\section{Proof of {\cref{prop:general_formula}}}\label{app:preliminary_results}

In this section we provide a proof for the building blocks of the main results of this work. Note that, what follows hinge on the structure of the circuit $\Phi_\theta$ provided in the main text, which we recall is composed of $L$ layers of interleaved unitary and noise channels 
\begin{equation}
    \Phi_{\theta} = \mathcal{U}_{\theta_{L+1}} \circ \mathcal{E}_{L}\circ\mathcal{U}_{\theta_{L-1}} \dots \circ  \mathcal{E}_{1}\circ\mathcal{U}_{\theta_{1}}.
\end{equation} 
In particular, we assume that $\mathcal{U}_\theta: \rho \mapsto U_\theta \rho U_\theta^\dagger$, where $U_\theta = \bigotimes_m U^{(m)}_{\theta_m}$ is a \emph{local 2-design} for the system. This statement is more precisely captured here. 
\begin{definition}(Local design)
Given a unitary ensemble $\{U_\theta\}_{\theta\in \Theta}$ with a given probability distribution over the parameter space $\Theta$, we say it forms a local $t$-design for the system if each element is factorized with respect to the partition, i.e. $U_\theta = \bigotimes_m U^{(m)}_{\theta_m}$ each acting solely on $\mathcal{H}_m$, and additionally
\begin{equation}
    \int_\Theta d\theta U^{(m) \otimes t }_\theta \otimes U^{(m)* \otimes t}_\theta = \int_{V \in U(d_m)} d\mu(V) V^{(m) \otimes t}\otimes V^{(m)* \otimes t}
\end{equation}
where the second integral is performed with respect to the Haar measure.
\end{definition}
With this notation in place, we are ready to start. The first Lemma provides a formula for the expectation value of a circuit in the aforementioned class, showing that they form global 1-designs.
\begin{lemma}[Global 1-design]\label{lemma:expval}
Let $A, B \in \BH$, and let $\{U_\theta\}_{\theta \in \Theta}$ be a unitary ensemble forming a local 1-design. Then
\begin{equation}
    \mathbb{E}_\theta \left\{ \Trb{ AU_\theta^{\dagger} B U_\theta}\right\} = \frac{\Trb{A}\Trb{B}}{d}
\end{equation}
\end{lemma}

\begin{proof}
Let $\{P_j\}_j$ be a local basis for the system, and consider the respective decompositions of $A$ and $B$, namely $A=\sum_{i} a_i P_i $ and  $B=\sum_{j} b_j P_j$. Note that, each component $a_i$ is defined as ${a_i = \Trb{P_iA}}$, and consequently $a_0 = \frac{1}{\sqrt{d}}\Trb{A}$ (respectively for $B$). Then we have the following chain of equalities:
\begin{equation}
\begin{split}
    &\int_{\Theta} \prod_{m=1}^{M} d\theta \Trb{ AU_\theta^{\dagger} BU_\theta}=\\
    = &\sum_{i,j} a_i b_j \int_\Theta \prod_{m=1}^{M} d\theta \Trb{ P_i U_\theta^{\dagger} P_j U_\theta} \\
    = &\sum_{i,j} a_i b_j \prod_{m=1}^{M} \int_{U\in U(d_m)} d\mu(U) \Trb[m]{ P^{(m)}_{i_m} U^{\dagger} P^{(m)}_{j_m} U}\\
    = &\sum_{i,j} a_i b_j \prod_{m=1}^{M} \frac{1}{d_m}\Trb[m]{P^{(m)}_{i_m}} \Trb[m]{P^{(m)}_{j_m}} \\
    = &\sum_{i,j} a_i b_j \prod_{m=1}^{M} \delta_{0,i_m}\delta_{0,j_m} = a_0 b_0
\end{split}
\end{equation}
where $\Tr_m$ denotes the partial trace over the $m$-th subsystem.
\end{proof}

Regarding the second moment, it can be computed using the following Lemma.

\begin{lemma}\label{lemma:weingarten}
Let $\{P_j\}$ be a local basis and let $\{U_\theta\}_{\theta \in \Theta}$ be a unitary ensemble forming a local 2-design. Then
\begin{equation}
\begin{split}
    &\mathbb{E}_\theta \left\{  \Trb{ P_i U_\theta^{\dagger} P_jU_\theta} \Trb{ P_k U_\theta^{\dagger} P_lU_\theta} \right\}= \\ = &\delta_{ij}\delta_{kl} \prod_{m=1}^M\left(\delta_{0i_m}\delta_{0j_m} + (1-\delta_{0i_m})(1-\delta_{0j_m})\frac{1}{d^2_m-1}\right)
\end{split}
\end{equation}
\end{lemma}

\begin{proof}
To prove this, we make use of the following result of Weingarten Calculus

\begin{equation}
\label{app:eq:weingarten}
\begin{split}
&\int_{U \in U(d)} d\mu(U) \Trb{AU^\dagger BU} \Trb{CU^\dagger DU} \\
&= \frac{1}{d^2-1} \left(\Trb{A}\Trb{B}\Trb{C}\Trb{D}+\Trb{AC}\Trb{BD}\right) \\
&+ \frac{1}{d(d^2-1)}\left(\Trb{AC}\Trb{C}\Trb{D}+\Trb{A}\Trb{C}\Trb{BD}\right)
\end{split}
\end{equation}

Based on \cref{app:eq:weingarten}, the result follows from direct integration:

\begin{equation}
\begin{split}
\prod_{m=1}^{M} &\int_{U\in U(d_m)} d \mu(U) \Trb{ P^{(m)}_{i_m} U^\dagger P^{(m)}_{j_m}U } \times\\
&\times\Trb{ P^{(m)}_{k_m} U^\dagger P^{(m)}_{l_m}U }\\
= \prod_{m=1}^M &\frac{1}{d_m^2-1} \Big(\Trb{P^{(m)}_{i_m}}\Trb{P^{(m)}_{j_m}}\Trb{P^{(m)}_{k_m}}\Trb{P^{(m)}_{l_m}}\\
+&\Trb{P^{(m)}_{i_m}P^{(m)}_{k_m}}\Trb{P^{(m)}_{j_m}P^{(m)}_{l_m}}\Big) \\
- &\frac{1}{d_m(d_m^2-1)}\Big(\Trb{P^{(m)}_{i_m}P^{(m)}_{k_m}}\Trb{P^{(m)}_{j_m}}\Trb{P^{(m)}_{l_m}}\\
&+\Trb{P^{(m)}_{i_m}}\Trb{P^{(m)}_{k_m}}\Trb{P^{(m)}_{j_m}P^{(m)}_{l_m}}\Big)\\
=\prod_{m=1}^M& \frac{1}{d_m^2-1} (d_m^2\, \delta_{0i_m}\delta_{0j_m}\delta_{0k_m}\delta_{0l_m}\\
&+ \delta_{i_m k_m}\delta_{j_m l_m} - \delta_{0i_m}\delta_{0k_m}\delta_{j_m l_m} - \delta_{i_m k_m}\delta_{0j_m}\delta_{0l_m})\\
=\prod_{m=1}^M& \delta_{0i_m}\delta_{0j_m}\delta_{0k_m}\delta_{0l_m}\\ 
&+\frac{1}{d_m^2-1} (\delta_{{i_m} {k_m}} - \delta_{0i_m}\delta_{0k_m}) (\delta_{{j_m} {l_m}} - \delta_{0j_m}\delta_{0l_m})\\
=\,&\delta_{ik}\delta_{jl} \prod_{m=1}^M \left(\delta_{0i_m}\delta_{0j_m} + (1-\delta_{0i_m})(1-\delta_{0j_m})\frac{1}{d^2_m-1}\right)
\end{split}
\end{equation}
\end{proof}

In particular, \cref{lemma:weingarten} can be used to compute the variance of loss functions computed as expectation values $\Loss = \Trb{\mathcal{U}_\theta(A)B}$, i.e. in the absence of intermediate channels $\mathcal{E}_l$.

\begin{proposition}\label{prop:warm_up}
Let $A, B \in \BH$, and let $\{U_\theta\}_{\theta \in \Theta}$ be a unitary ensemble forming a local 2-design. Then

\begin{equation}
\mathbb{E}_\theta \left\{\Trb{ AU_\theta^{\dagger} B U_\theta}^2 \right\} = \left(\ell_A, \ell_B\right),
\end{equation}

where $(\cdot,\cdot)$ is the scalar product defined in \cref{eq:scalar_prod}.
\end{proposition}

\begin{proof}
Let $\{P_j\}$ be a local basis for the system, and consider the respective decompositions of $A$ and $B$. By \cref{lemma:weingarten} we have

\begin{equation}\label{eq:lemma_weingarten0}
\begin{split}
    \mathbb{E}_\theta &\left\{\Trb{ AU_\theta^\dagger B U_\theta}^2 \right\} = \\&\sum_{i,j} a_i^2b_j^2 \prod_{m=1}^M \left(\delta_{0i_m}\delta_{0j_m} + (1-\delta_{0i_m})(1-\delta_{0j_m})\frac{1}{d^2_m-1}\right)
\end{split}
\end{equation}

In the following, it will be convenient to recast the product on the right-hand side of \cref{eq:lemma_weingarten0} into the equivalent formulation

\begin{equation}
    \begin{split}
    &\prod_{m=1}^M\left(\delta_{0i_m}\delta_{0j_m} + (1-\delta_{0i_m})(1-\delta_{0i_m})\frac{1}{d^2_m-1}\right) =\\ &= \sum_{\kappa \in \{0,1\}^{M}} \frac{1}{d_\kappa} \prod_{m=1}^M (\delta_{0i_m}\delta_{0j_m})^{1-\kappa_m}(1-\delta_{0i_m})^{\kappa_m}(1-\delta_{0j_m})^{\kappa_m}
    \end{split}
\end{equation}

where the binary vectors $\kappa\in\{0,1\}^M$ identify all possible sets $S_\kappa$ introduced in \cref{app:loc_properties} and $d_\kappa = \prod_{m=1}^M({d_m^2-1})^{\kappa_m}$. Putting it back into \cref{eq:lemma_weingarten0} we get

\begin{equation}
\begin{split}
     &\sum_{\kappa \in \{0,1\}^{M}} \frac{1}{d_\kappa} \sum_{i,j} a_i^2 b_j^2 \prod_{m=1}^M (\delta_{0i_m}\delta_{0j_m})^{1-\kappa_m}\times\\
     &\times(1-\delta_{0i_m})^{\kappa_m}(1-\delta_{0j_m})^{\kappa_m}\\
     = &\sum_{\kappa \in \{0,1\}^{M}} \frac{1}{d_\kappa}\left(\sum_i a_i^2 \prod_{m=1}^M \delta_{0i_m}^{1-\kappa_m}(1-\delta_{0i_m})^{\kappa_m}\right)\times\\
     &\times\left(\sum_j b_j^2 \prod_{m=1}^M \delta_{0j_m}^{1-\kappa_m}(1-\delta_{0j_m})^{\kappa_m}\right)\\
     = &\sum_{\kappa \in \{0,1\}^{M}} \frac{1}{d_\kappa}\left(\sum_i a_i^2 \Tilde{\delta}_{i, \kappa}\right)\left(\sum_j b_j^2 \Tilde{\delta}_{j, \kappa}\right)\\
     = &\sum_{\kappa \in \{0,1\}^{M}} \frac{(\ell_A)_\kappa (\ell_B)_\kappa}{d_\kappa} = (\ell_A, \ell_B)
\end{split}
\end{equation}
from which the Proposition follows.

\end{proof}

This can be extended to the general case introducing the action of the intermediate channels $\mathcal{E}_l$, and in particular we get the following Proposition.

\begin{proposition}\label{app:prop:general_formula}
Let $A, B \in \BH$ and $\Lambda:\BH \to \BH$ be a linear map. Furthermore, let $\{U_{\theta_1}\}_{\theta_1 \in \Theta}$ and $\{V_{\theta_2}\}_{\theta_2 \in \Theta}$ be independent, unitary ensembles each forming a local 2-design. Then

\begin{equation}
    \mathbb{E}_{\theta_1,\theta_2} \left\{ \Trb{AU^\dagger_{\theta_1}\Lambda(V_{\theta_2}BV^\dagger_{\theta_2})U_{\theta_1}}^2 \right\} = (\ell_A,T\ell_B)
\end{equation}

where $T$ is the locality transfer matrix associated to $\Lambda$.

\end{proposition}

\begin{proof}
Let $\Tilde{B}_{\theta_2} = \Lambda(V_{\theta_2} B V_{\theta_2})$. By \cref{app:prop:general_formula}, we have

\begin{equation}
\begin{split}
    &\mathbb{E}_{\theta_2} \mathbb{E}_{\theta_1}\left\{ \Trb{ A U^\dagger_{\theta_1} B_{\theta_2}U_{\theta_1}}^2 \right\}\\
    &= \mathbb{E}_{\theta_2} \left\{ ( \ell_A,\ell_{\Tilde{B}_{\theta_2}}) \right\} \\
    &= \sum_{\kappa \in \{0,1\}^{M}} \frac{1}{d_\kappa} (\ell_{A})_\kappa \,\mathbb{E}_{\theta_2} \left\{ (\ell_{\Tilde{B}_{\theta_2}} )_\kappa \right\}
\end{split}
\end{equation}

Expanding the definition on the last term with respect to the local basis $\{P_j\}_j$, ad applying again \cref{app:prop:general_formula} we get

\begin{equation}
\begin{split}
&\mathbb{E}_{\theta_2} \left\{ (\ell_{\Tilde{B}_{\theta_2}})_\kappa \right\} = \sum_i \mathbb{E}_{\theta_2} \left\{\Trb{P_{i}\Tilde{B}_{\theta_2}}^2 \right\} \Tilde{\delta}_{i, \kappa}\\ 
&= \sum_i \mathbb{E}_{\theta_2} \left\{\Trb{\Lambda^\dagger(P_{i})V^\dagger_{\theta_2} BV_{\theta_2}}^2 \right\} \Tilde{\delta}_{i, \kappa}\\
&= \sum_{\lambda \in \{0,1\}^{M}} \frac{1}{d_\lambda} \sum_{i, j,k} \Trb{P_i\Lambda(P_j)}^2 \Trb{P_kB}^2 \Tilde{\delta}_{j, \lambda} \Tilde{\delta}_{k, \lambda} \Tilde{\delta}_{i, \kappa}\\
&= \sum_{\lambda \in \{0,1\}^{M}} \left( \frac{1}{d_\lambda} \sum_{i, j} \Trb{P_i\Lambda(P_j)}^2 \Tilde{\delta}_{i, \kappa} \Tilde{\delta}_{j, \lambda}\right)\times\\
&\times\left(\sum_k\Trb{P_kA}^2 \Tilde{\delta}_{k, \lambda} \right)\\
&= \sum_{\lambda \in \{0,1\}^{M}} T_{\kappa,\lambda} (\ell_B)_\lambda = (T \ell_B)_\kappa
\end{split}
\end{equation}

where $\Lambda^\dagger$ is the Hermitian adjoint of $\Lambda$ with respect to the Hilbert-Schmidt scalar product.
\end{proof}

Iterated application of \cref{app:prop:general_formula} for an initial state $\rho$, an observable $H$, and a general intermediate quantum channel $\mathcal{E}$, yields \cref{eq:general_formula} of the main text.

\section{Proof of \cref{thm:main}}\label{appendix:proof_main}

The proof of \cref{thm:main} is based on the characterization of the general LTM for the Hermitian adjoint $\mathcal{E}^\dagger$ of arbitrary quantum channel. To do so, several aspects of non-negative matrix theory, as well as the contractivity properties of $\mathcal{E}^\dagger$ are key. For sake of clarity and completeness we recall them in the following.

\subsection{Preliminaries}

In this section, we start with preliminary concepts and definitions involving non-negative matrices, and then we recall some well known facts and definitions about operator norms and quantum channel contractivity.

\subsubsection{Elements of non-negative matrix theory}

In this section, we briefly recap on the main results on non-negative matrix theory useful in the proof of \cref{thm:main} of the main text. For a complete discussion and proofs of the cited results, we refer the interested reader to Refs.~\cite{Seneta06, Meyer00}. Let's start by the definition of non-negative matrix.
\begin{definition}[Non-negative matrix]
    A $n\times n$ matrix $T$ is said to be non-negative if each entry $(T)_{ij}\geq 0$.
\end{definition}
The general behaviour of non negative matrices can vary greatly, but there is a class of matrices, called \emph{irreducible}, which have very informative spectral properties.
\begin{definition}[Irreducible matrix] A $n \times n$ non-negative matrix $T$ is said to be \emph{irreducible} if for two arbitrary indices $i,j =1,...n$, there exist $l=l(i,j)\in \mathbb{N}$ such that $(T^l)_{ij}>0$ . Moreover, we will say that $T$ has period $p$, where $p$ is the greatest common divisor of all $l(i,i)$ that satisfy $(T^l)_{ii}>0$ $\forall i$.
\end{definition}
Equivalently, if we introduce the graph $\mathcal{G}_T$ whose adjacency matrix is $T$, then it can be shown that $T$ is irreducible if and only if $\mathcal{G}_T$ is strongly connected, and that the period $p$ reduces to the great common divisor of the lengths of all closed directed paths in $\mathcal{G}_T$ \cite{Seneta06}.  Furthermore, it will be useful in the following to distinguish two classes of irreducible matrices, namely \emph{cyclic} (or \emph{periodic}) and \emph{primitive} (or \emph{aperiodic}), which are characterized as having period $p>1$ and $p=1$ respectively. 
\\ \\
One of the main results involving irreducible matrices is the celebrated \emph{Perron-Frobenius} theorem, which characterizes the spectral properties of this class. We recall it here for convenience.

\begin{theorem}[Perron-Frobenius]
Let $T$ be a $n\times n$ non-negative, irreducible matrix. Then there exists an eigenvalue $r$ of $T$, with corresponding right and left eigenvectors $v$, $w$ such that:
    \begin{enumerate}[(a)]
        \item $r\in \mathbb{R}, r>0$ and is a simple root of the characteristic polynomial,
        \item both $w$,$v$ are the only eigenvectors that have strictly positive components, i.e. $v_i,w_i>0 \, \forall i=1,...n$,
        \item $v$ and $w$ are unique up to a scalar multiple, and hence can be taken to be normalized, i.e. $v^tw=1$,
        \item $r \geq |\lambda|$, for all eigenvalue $\lambda$ of $T$,
    \end{enumerate}
where $r$ is called the Perron-Frobenius eigenvalue and $P = wv^t$ the Perron projector. Moreover, if $T$ is also aperiodic, then we have the more restrictive
    \begin{enumerate}[(a'), start=4]
        \item $r>|\lambda|$, for all eigenvalue $\lambda \neq r$ of $T$,
    \end{enumerate}
\end{theorem}
Another important result is the so called subinvariance theorem, which is a useful tool to bound the value of $r$ for a given irreducible matrix.
\begin{theorem}[Subinvariance Theorem]
    Let $T$ be a $n\times n$ non-negative, irreducible matrix, $s>0$ and $y$ be a $n$-dimensional row vector such that each component $y_i\geq0$ and satisfying
    \begin{equation}
        y^tT \leq sy^t
    \end{equation}
    component-wise. Then $y_i>0 \forall i$, and $s\geq r$. Moreover, equality holds if and only if $s=r$.
\end{theorem}

Finally, the last result allows us to exploit the knowledge of the dominant eigenvalue to determine the asymptotic properties of $T^L$.

\begin{theorem}[Asymptotic behaviour of irreducible matrices] Let $T$ be a $n\times n$ non-negative, irreducible matrix. Then the Cesàro average of $T$ converges, and we have 
\begin{equation}
    \lim_{L\to\infty} \frac{1}{L}\sum_{l=1}^{L} T^l/r^l = P
\end{equation}
Moreover, if $T$ is also aperiodic, then limit of $T^L/r^L$ converges, and we have
\begin{equation}
    \lim_{L\to\infty} T^L/r^L = P
\end{equation}
where $r$ and $P$ are the Perron-Frobenius eigenvalue and Perron projector respectively.
\end{theorem}

Despite being less structured, it is a well known fact that general non-negative matrices can be cast to a canonical block upper triangular form, where all blocks in the diagonal are \emph{irreducible} simply by means of a permutation matrix, i.e. by a relabelling of the basis elements. In particular, concerning the diagonal, irreducible blocks appearing in such decomposition, we will use the term \emph{essential} when referring to the blocks such that all $(T)_{i,j} = 0$ for all columns apart from the block itself, and \emph{inessential} otherwise. In terms of the graph $\mathcal{G}_T$, this distinction is readily understood. As discussed above, irreducible blocks correspond to strongly connected components, and consequently essential blocks are strongly connected components which do not have edges connecting vertices in it to vertices pertaining to other components. More simply, we can describe essential components as those whose edges ``do not lead outside". In what follows, we name by $T_z$ all essential, irreducible components, and we group into a single block $Q$ all inessential components. The blocks $R_z$, appearing on top of the block $Q$, represent the collection of edges coming from inessential components and leading to essential ones. A graphical summary is depicted in \cref{eq:T_canonical_form}.

\begin{equation}\label{eq:T_canonical_form}
        T = 
    \begin{pmatrix}
    \,\tikz{\node[draw]{$T_0$}} & & & & \tikz{\node[draw, minimum width=1.2cm]{$R_0$}}\,\\
    & \tikz{\node[draw, minimum width=0.8cm, minimum height=0.8cm] {$T_1$}}& & & \tikz{\node[draw, minimum width=1.2cm, minimum height=0.8cm]{$R_1$}}\\
    & & \ddots & & \vdots\\
    & & & \tikz{\node[draw, minimum width=0.8cm, minimum height=0.8cm] {$T_z$}} & \tikz{\node[draw, minimum width=1.2cm, minimum height=0.8cm]{$R_z$}}\\
    & & & & \tikz{\node[draw, minimum width=1.2cm, minimum height=1.2cm] {$Q$}}\\
    \end{pmatrix}
\end{equation}
This in particular allows us to apply the results of this section also to more generic matrices, such as LTMs, which in general are not irreducible. 

\subsubsection{Useful results on quantum channel}

In this section, we briefly introduce some relevant properties of completely positive (CP) maps. These will be especially useful in the trace preserving case (CPTP), i.e. quantum channels, and the unital case (CPU), i.e. the corresponding adjoint action with respect to the Hilbert-Schmidt scalar product. In what follows, we will denote by $\mathbb{M}_n$ the space of $n\times n$ matrices, which we can endow with a norm as follows.
\begin{definition}(Schatten norm)\label{def:schatten-norms} Given $A \in \mathbb{M}_n$ and $p\in [1,\infty]$, we define the Schatten $p$-norm  $\|\cdot\|_p :\mathbb{M}_n \to \mathbb{R}$ as
\begin{equation}
    \|A\|_p = \Trb{\left(\sqrt{A^\dagger A}\right)^p}^{1/p}.
\end{equation}
\end{definition}
A crucial property of Schatten norms is H\"older inequality, namely $|\Trb{A^\dagger B}| \leq \|A\|_p\|B\|_q$, $\forall A, B \in \mathbb{M}_n, \forall p,q$ s.t. $1/p+1/q = 1$. As a special case, for $p=2$ we get back the Hilbert-Schmidt norm, and H\"older inequality reduces to the Cauchy-Schwarz inequality.

As a direct consequence of \cref{def:schatten-norms}, an induced norm on linear operators acting on $\mathbb{M}_n$ can be defined.
\begin{definition}(Induced norm)\label{def:induced_norm} Given a linear operator $\Lambda:\mathbb{M}_n\to\mathbb{M}_n$ we define the induced $p\to q$ norm as
\begin{equation}
    \|\Lambda\|_{p\to q} := \sup_{A : \|A\|_p=1} \|\Lambda(A)\|_q.
\end{equation}
\end{definition}
Depending on the properties of such induced norms, we might refer to the map $\Lambda$ as contractive or strictly contractive. In particular, we will use the following definitions.
\begin{definition}(Contractivity of linear maps)\label{def:contractivity} A linear operator $\Lambda:\mathbb{M}_n\to\mathbb{M}_n$ is said to be \emph{contractive} with respect to the $p$ norm if $\|\Lambda\|_{p\to p} \leq 1$, and similarly to be \emph{strictly contractive} if $\|\Lambda\|_{p\to p} < 1$.
\end{definition}
Linear maps that are also CPTP are known to always be contractive with respect to the $1$-norm \cite{Raginsky02,Holevo01}, but are in general not contractive for other $p$-norms. This property, together with H\"older's inequality, allow putting an upper bound on the value of the variance of an arbitrary, layered quantum circuit. 
\begin{lemma}[H\"older's inequality for variances]\label{lemma:holder_for_variance}
   Let $A,B \in \mathbb{M}_n$ be hermitian matrices, and $\Phi_\theta:\mathbb{M}_n \to \mathbb{M}_n$ be a parameterized quantum channel. In particular, consider the $L$ layered map  of type \cref{eq:circuit_channel}. Then
   \begin{equation}
    \left| \Trb{\Phi_\theta^L(A) B}\right|\leq  \|\mathcal{E}\|^{L}_{p\to p}\|\|A\|_p\|B\|_q \;\; \forall \theta \in \Theta
   \end{equation}
   with $1/p+1/q = 1$. As a special case, if $A=\rho$ and $B=H$ are a density operator and an observable respectively, then the variance can be upper bounded by  $\Varr^L \leq \|H\|_\infty^2 \forall L$.
\end{lemma}
\begin{proof}
The result is the direct consequence of H\"older's inequality and contractivity of quantum maps. In particular, we have the following chain of inequalities:
\begin{equation}
\begin{split}
    \left| \Trb{\Phi^L_{\theta}(A)B} \right| &\leq \|\Phi^L_{\theta}(A)\|_p\|B\|_q \\
    &= \|U_{\theta_L} \mathcal{E}\left(\Phi^{L-1}_{\theta}(A) \right) U^\dagger_{\theta_L}\|_p \|B\|_q\\
    &= \|\mathcal{E}\left(\Phi^{L-1}_{\theta}(A) \right)\|_p \|B\|_q\\
    &\leq \|\mathcal{E}\|_{p\to p} \|\Phi^{L-1}_{\theta}(A)\|_p \|B\|_q  \;\; \forall \theta \in \Theta
\end{split}
\end{equation}
By iterative application of this procedure, one can get the result. The final remark holds due to the contractivity of quantum channels, choosing $p=1$ and $q=\infty$, and noting that $\|\rho\|_1 \leq 1$ for all density matrices. 
\end{proof}
Specific classes of quantum channels can be shown to be contractive with respect to a wider variety of norms. In particular, we have that, for $p\geq 2$, \emph{unitality} of the map is a necessary and sufficient condition for contractivity. (see Theorem II.4 in \cite{PerezGarcia06}). Within unital channels, unitary transformation $\mathcal{U}$ always saturate the bound, as they have the additional property of being norm-preserving, i.e. $\|\mathcal{U}(A)\|_p = \|U^\dagger AU\|_p = \|A\|_p$.
Finally, if we reduce the action of the channel to the subset $\mathbb{H}_0 \subset \mathbb{M}_n$ of Hermitian, traceless matrices, then the unitality property is no longer a necessary condition for contractivity. Indeed, any single qubit channel $\mathcal{N}$ is contractive in this setting, i.e. $\|\mathcal{N}|_{\mathbb{H}_0}\|_{p\to p}\leq 1 \forall p$ if $n=2$. For single qubit channels we can even be more explicit, as showed in the following Lemma.
\begin{lemma}[Single qubit channel normal form]\label{lemma:normal_form} Let $\{P_i\}$ be the normalized (with respect to the Schatten $2$-norm) Pauli basis of $\mathbb{M}_2$, and  $\mathcal{N}$ be a single qubit channel. Then there exist unitary matrices $U, V$ such that $\mathcal{N}'(\cdot) = U^\dagger\mathcal{N}(V^\dagger \cdot V)U$ satisfies
\begin{equation}\label{eq:normal_form}
    \mathcal{N}'(P_0) = P_0 + \sum_{i>0} t_i P_i, \quad \mathcal{N}'(P_i) = \lambda_i P_i \, \forall i>0
\end{equation}
where $\sum_{i>0} (t_i + \lambda_i \alpha_i)^2 \leq 1$, $\forall \alpha_i \in \mathbb{R}$ s.t. $\sum_{i>0} \alpha_i^2 \leq 1$.
\end{lemma}
\begin{proof}
    It has been shown in \cite{King00,Ruskai02} that any single qubit quantum channel can be cast in the canonical form of \cref{eq:normal_form} by means of a change of basis. Furthermore, the constraint on the parameters follow, analogously to \cite{Mele24}, by considering that any single qubit state must have bounded purity, namely $\Trb{\rho^2} \leq 1$, and since $\mathcal{N}'$ is a channel, the same must hold for $\mathcal{N}'(\rho)$. Since any qubit state can be decomposed in terms of $\{P_i\}_i$ as $\rho = 1/\sqrt{2}P_0 + 1/\sqrt{2} \sum_{i>0} \alpha_i P_i$, $\alpha_i \in \mathbb{R}$, we get
    \begin{equation}
        \Trb{\rho^2} = \frac{1+\alpha_i^2}{2} \leq 1, \quad \Trb{\mathcal{N}'(\rho)^2} = \frac{1+\sum_{i>0} (t_i + \lambda_i \alpha_i)^2}{2} \leq 1
    \end{equation}
    respectively, which concludes the proof.
\end{proof}
Considering instead CPU maps, the most relevant result is Kadison-Schwarz inequality, which for our purposes, can be stated as follows.
\begin{theorem}[Kadison-Schwarz inequality \cite{Kadison52}]\label{thm:kadison_schwarz}
    Let $A, B \in \mathbb{H}_0$, and $\Lambda:\mathbb{M}_n\to\mathbb{M}_n$ be a CPU map. Then 
    \begin{equation}
        \Lambda(A)\Lambda(B) \leq \Lambda(AB). 
    \end{equation}
\end{theorem}
All these properties will be useful to characterize the spectral properties of interest of the LTM of quantum channels.

\subsection{Further characterizations of LTMs}

We now study the structure of the LTM of a general CPU map. Thanks to this analysis, we will be able to compute the limiting value $\Varr^\infty$ by describing the quantum circuit in the Heisenberg picture. Denoting by $T$ the resulting LTM, we start by computing the general form of integer powers $T^L$ of $T$.

\begin{lemma} [Limiting form of $T$]\label{lemma:T_lim}
    Let $\Lambda:\BH\to\BH$ be a CPU map and $T$ be the corresponding LTM. Then $T$ and $T^L$ take the form
\begin{equation}
\begin{split}\label{eq:canaonical_form_appendix}
    T = 
    \begin{pmatrix}
    \,\tikz{\node[draw]{$T_0$}} & & & & \tikz{\node[draw, minimum width=1.2cm]{$R_0$}}\,\\
    & \tikz{\node[draw, minimum width=0.8cm, minimum height=0.8cm] {$T_1$}}& & & \tikz{\node[draw, minimum width=1.2cm, minimum height=0.8cm]{$R_1$}}\\
    & & \ddots & & \vdots\\
    & & & \tikz{\node[draw, minimum width=0.8cm, minimum height=0.8cm] {$T_z$}} & \tikz{\node[draw, minimum width=1.2cm, minimum height=0.8cm]{$R_z$}}\\
    & & & & \tikz{\node[draw, minimum width=1.2cm, minimum height=1.2cm] {$Q$}}\\
    \end{pmatrix}, \\ \\
    T^L = \begin{pmatrix}
    \,\tikz{\node[draw]{$T_0^L$}} & & & & \tikz{\node[draw, minimum width=1.2cm]{$A_0^{(L)}$}}\,\\
    & \tikz{\node[draw, minimum width=0.8cm, minimum height=0.8cm] {$T_1^L$}}& & & \tikz{\node[draw, minimum width=1.2cm, minimum height=0.8cm]{$A_1^{(L)}$}}\\
    & & \ddots & & \vdots\\
    & & & \tikz{\node[draw, minimum width=0.8cm, minimum height=0.8cm] {$T_z^L$}} & \tikz{\node[draw, minimum width=1.2cm, minimum height=0.8cm]{$A_z^{(L)}$}}\\
    & & & & \tikz{\node[draw, minimum width=1.2cm, minimum height=1.2cm] {$Q^L$}}\\
    \end{pmatrix}
\end{split}
\end{equation}
up to a basis state index permutation, where each $T_z$ is an irreducible matrix, and $A^{(L)}_z = \sum_{l=0}^{L-1}T_z^lR_zQ^{L-1-l}$.
\end{lemma}
\begin{proof}
    We start by putting $T$ into the canonical form of \cref{eq:T_canonical_form}. In this form, the powers of the diagonal blocks $T_z^L$ are trivially the diagonal blocks of $T^L$. Instead, the result about $A^{(L)}$ follows by induction.
    In fact, both the base case and the inductive one follow from matrix multiplication rules, of $T\cdot T$ and $T^{L-1}\cdot T$ respectively. In particular we have
    \begin{equation}
    \begin{split}
        A_z^{(2)} &= T_zR_z + R_zQ\\
        A_z^{(L)} &= T_z^{L-1} R_z + A^{(L-1)}Q\\
        &=T_z^{L-1} R_z + \sum_{l=0}^{L-2}T_z^lR_zQ^{L-1-l} \\
        &= \sum_{l=0}^{L-1}T_z^lR_zQ^{L-1-l}\\
    \end{split} 
    \end{equation}

which gives the proposition.
\end{proof}

As already shown in \cref{lemma:holder_for_variance}, the value of the variance is upper bounded by $\|H\|_\infty^{2}$. Since by \cref{prop:general_formula}, this quantity is linked to $(\ell_\rho, T^L \ell_H)$, it is expected that the spectral radius $\rho(T)$ of $T$ is upper bounded by $1$. More specifically, we can prove the following Proposition. 

\begin{proposition}[Spectral radius of T]\label{prop:spectral_radius_bound}
Let $\Lambda:\BH\to\BH$ be a CPU map and $T$ be the corresponding LTM. Then $T$ is contractive in the sense of the spectral radius, i.e. $\rho(T) \leq 1$. Moreover, the component $Q$ is strictly contractive, namely $\rho(Q) < 1$.
\end{proposition}

\begin{proof}
The statement follows as a consequence of the Kadison-Schwarz inequality and the Subinvariance theorem. 

First note that, for a generic $\Lambda$, the structure of $Q$ is not as well-behaved as $T_z$, as $Q$ is not necessarily \emph{irreducible}. However, as any non-negative matrix, also $Q$ can be cast in canonical block upper triangular form by means of a basis permutation, where each diagonal block is irreducible.

\begin{equation}Q = 
\begin{pmatrix}
\,\tikz{\node[draw, minimum width=0.8cm, minimum height=0.8cm]{$Q_1$}} & \tikz{\node[draw, minimum width=0.8cm, minimum height=0.8cm]{*}} & \cdots &  \tikz{\node[draw, minimum width=0.8cm, minimum height=0.8cm]{*}}\,\\
& \tikz{\node[draw, minimum width=0.8cm, minimum height=0.8cm] {$Q_2$}}& & \tikz{\node[draw, minimum width=0.8cm, minimum height=0.8cm]{*}}\\
& & \ddots & \vdots\\
& & & \tikz{\node[draw, minimum width=0.8cm, minimum height=0.8cm] {$Q_k$}}\\
\end{pmatrix}
\end{equation}

With this in mind, we can study the spectral radius of $T$ and $Q$ in terms of the spectral radii of each block $T_z$ and $Q_k$, i.e. the corresponding Perron eigenvalues, since $\rho(T) = \max\{\max_z r_z, \max_k r_{Q_k}\}$ and $\rho(Q) = \max_k r_{Q_k}$.

By \cref{lemma:basis_independence_T}, we are free to choose the basis used to express the matrix $T$. In particular, we choose the normalized Pauli basis $\{P_i\}$, which, besides being a local basis for $\BH$, is also unitary up to normalization, namely $P_i^2=\mathds{1}/d \, \forall i$. Exploiting \cref{eq:alternate_LTM}, we can compute the column sum of $T$ as

\begin{equation}\label{eq:substochasticity}
\begin{split}
    &\sum_\kappa T_{\kappa, \lambda} = \sum_\kappa  \frac{1}{d_\lambda} \sum_{i,j} \Trb{P_i \Lambda(P_j)}^2 \deltat{i}{\kappa} \deltat{j}{\lambda}\\
    &= \frac{1}{d_\lambda} \sum_{i,j} \Trb{P_i \Lambda(P_j)}^2 \deltat{j}{\lambda} \left(\sum_\kappa \deltat{i}{\kappa} \right)\\
    &=\frac{1}{d_\lambda} \sum_{i,j} \Trb{P_i \Lambda(P_j)}^2 \deltat{j}{\lambda} =  \frac{1}{d_\lambda} \sum_{j} \Trb{ \Lambda(P_j)^2} \deltat{j}{\lambda} \\
    &\leq \frac{1}{d_\lambda} \sum_{j} \Trb{ \Lambda(P_j^2)} \deltat{j}{\lambda} = \frac{1}{d_\lambda} \sum_{j} \deltat{j}{\lambda} = 1
\end{split}
\end{equation}

where the third and last equality follow from \cref{lemma:indicator_func}, the inequality is Kadison-Schwarz and the second to last equality is the unitary property of the basis. If $\Lambda$ is unitary, the inequality is saturated, and $T$ becomes a \emph{stochastic} matrix. In general, \cref{eq:substochasticity} show \emph{sub-stochasticity} of T. Indeed, this condition can be recast in vector form as $v^tT \leq v^t$, where $v_\kappa = 1 \forall \kappa$. In particular, this holds for all irreducible blocks in the diagonal, which by the Subinvariance theorem implies $r_z \leq 1$ and $r_{Q_k} \leq 1$, giving $\rho(T) \leq 1$. Focusing on $Q$, we observe that, by definition, each irreducible block $Q_k$ is \emph{inessential}, i.e. is connected to some other block. In terms of the matrix $T$, this means that, considering the columns involving $Q_k$, there is always an index $\lambda$ in the support of $Q_k$ such that
\begin{equation}
    \sum_\kappa (T-Q_k)_{\kappa, \lambda} > 0,
\end{equation}

which implies that $\exists \lambda$ s.t. $\sum_\kappa (Q_k)_{\kappa, \lambda} < 1$. Written in matrix form, this reads $v^t Q_k \leq v^t$, $v^t Q_k \neq v^t$, which by the Subinvariance theorem, implies $r_{Q_k} \neq 1$. Putting everything together, one gets $\rho(Q) < 1$. 

\end{proof}

When analysing the single irreducible components $T_z$, we can be more specific, and find an equivalence between the value of the column-sum of the block and the value of the corresponding spectral radius. This is especially useful in the computation of the dominant eigenvectors, which is explicitly stated in the following Corollary.

\begin{corollary}\label{cor:column_sum}
    Let T be a LTM and $T_z$ be an irreducible block, then $\rho(T_z) = 1 \Leftrightarrow \sum_\kappa (T_z)_{\kappa, \lambda} = 1$, or equivalently $v_z^t T_z = v_z^t$, where $(v_z)_\kappa = 1 \forall \kappa$ is the left eigenvector of the dominant eigenvalue.
\end{corollary}

\begin{proof}
    The result follows from the same proof strategy as above, and is a direct consequence of the Subinvariance theorem.
\end{proof}

Intuitively, the blocks which are out of such hypothesis won't contribute to the large $L$ limit, and indeed the contribution of the $Q$, $T_z$ and $A_z^{(L)}$ is bounded to decay exponentially in the number of layers.

\begin{proposition}\label{prop:vanishing_terms}
Let $T$ be a LTM, and let $T_z$ be an irreducible block with $r_z < 1$. Then, as $L\to \infty$, $\|T_z^L\| \to 0$ and $\|A^{(L)}_z\| \to 0$ exponentially fast for any matrix norm $\|\cdot\|$. Similarly, also $\|Q^L\| \to 0$.
\end{proposition}

\begin{proof}
The proposition can be proven using Gelfand's formula. In particular since $\lim_{L\to\infty} \|T_z^L\|^{1/L} =  r_z$, we can always bound $\|T_z^L\| \leq K \tau^L$, for some constant $K>0$ and $\tau = r_z +\epsilon < 1$ for an arbitrarily small $\epsilon$. In the same way, by \cref{prop:spectral_radius_bound} a similar result can be obtained for Q. Finally, the absorption term $A_z^{(L)}$ can also be bounded using \cref{lemma:T_lim}. In that case we have
\begin{equation}
\begin{split}
        \|A_z^{(L)}\| &\leq \sum_{l=0}^{L-1} \|T_z^L\|\|R_z\|\|Q^{L-1-l}\|\\
        &\leq \|R_z\| K_T K_Q \tau^l \kappa^{L-1-l} \leq K_A\alpha^L
\end{split}
\end{equation}

with some constant $K_A>0$ and $\alpha = \max\{\kappa, \tau\}<1$. This can be obtained again using Gelfand's formula on both $T_z$ and $Q$, and by sub-additivity and sub-multiplicativity of the matrix norm $\|\cdot\|$.
\end{proof}

\subsection{Proof of \cref{thm:main}}

As stated in the main text, the limiting value of quantum circuits of type \cref{eq:circuit_channel} is obtained by studying the spectral properties of the LTM of the intermediate channel in the Heisenberg picture. In particular, the previous discussion  suggests a limiting value for the variance of the form
\begin{equation}\label{eq:v_infty_general_app}
    \Varr^\infty = \sum_{z\,|\,\rho(T_z)=1} (\ell_\rho,w_z) (\ell_H)_z + (\ell_\rho,w_z) (A_z \ell_H)
\end{equation}

with some normalized, strictly positive vector $w_z$, and absorption matrices $A_z$. Indeed, the following shows that this is the case.

\begin{theorem}[Deep circuit variance]\label{thm:main_appendix}
Let $\rho, H \in \BH$ and let $\Phi_\theta$ a be layered quantum channel as in described in the main text. Then the Cesàro average of $\Varr^L$ converges to \cref{eq:v_infty_general_app}, and we have
\begin{equation}
\label{eq:cesaro_limit_app}
    \left|\frac{1}{L} \sum_{l=0}^L \Varr^l -\Varr^\infty \right| \in O\left(e^{-\beta L} \|H\|^2_2\right),
\end{equation}
for some constant $\beta > 0$. Additionally, if all essential blocks are aperiodic, then $\Varr^L$ is convergent, and we have
\begin{equation}
\label{eq:simple_limit_app}
    \left|\Varr^L -\Varr^\infty \right| \in O\left(e^{-\beta L} \|H\|^2_2\right),
\end{equation}
where the right eigenvector $w_z$ of $T_z$ is a strictly positive vector, i.e. $(w_z)_\kappa > 0 \, \forall \kappa$, and $A_z = R_z(\mathds{1}-Q)^{-1}$ are the absorption coefficients of each essential block. 
\end{theorem}

\begin{proof}
Thanks to \cref{prop:vanishing_terms}, only irreducible components with $\rho(T_z)=1$ will contribute to the limit, so we can restrict our analysis to those alone.
Consider then an irreducible block $T_z$ with unit spectral radius, and of period $d$. While the full version of Perron-Frobenius theorem does not directly apply to $T_z$, it is a well-known result of non-negative matrix theory \cite{Seneta06,Meyer00} that the matrix $T_z^d$ can be cast to a block diagonal form by a permutation, with irreducible and \emph{aperiodic} blocks, for which we can apply it. However, it is crucial to notice that while $\lim_{N\to \infty} T_z^{dN} = T_z^{(d\infty)}$ exists, this does not imply that $\lim_{L\to \infty} T_z^{L}$ does. In fact, different sub-sequences might have different limiting values, and in particular $\lim_{N\to\infty} T_z^{dN+m} = T_z^{(d\infty)}T^m$, which is different for all $m=0,...d-1$. In the periodic scenario then, $T_z^L$ does not have a limit, and the only convergent quantity is the \emph{Cesàro average}, i.e.
\begin{equation}
    P_z = \lim_{L\to\infty} \frac{1}{L}\sum_{l=1}^L T_z^l = T_z^{(d \infty)}\frac{1}{d}\sum_{m=0}^{d-1}T_z^m
\end{equation}

where $P_z = w_zv_z^{t}$ can be shown to be the Perron projector associated to $T_z$ \cite{Seneta06,Meyer00}. Despite more cumbersome, a totally analogous approach allows determining the limiting values of $A^{(dN+m)}$ as well, as shown in the following Proposition.

\begin{proposition}
    Given a LTM with a periodic irreducible block $T_z$ of period $d$, then
    \begin{equation}
        \lim_{N\to\infty} A^{(dN+m)} = T^{(d\infty)} A^{(m)} + A^{(d\infty)} Q^m
    \end{equation}
where $A^{(d\infty)} = \sum_{m=0}^{d-1} T^{(d\infty)}A^{(d)}(\mathds{1}-Q^d)^{-1}$ and $A^{(0)}=0$.
\end{proposition}

\begin{proof}
Starting from the definition of $A^{(l)}$, we can first find the limiting value $A^{(d\infty)}$ of $A^{(dN)}$:

\begin{equation}\label{thm:main_dim:prop:chain}
\begin{split}
    A^{(dN)} &= \sum_{l=0}^{Nd-1}T_z^lR_z Q^{dN-1-l}\\
    &=\sum_{m=0}^{d-1}\sum_{n=0}^{N-1} T_z^{nd+m}R_zQ^{(N-1)d-nd+d-1-m}\\
    &=\sum_{m=0}^{d-1}\sum_{n=0}^{N-1} T_z^{nd}T_z^{m}R_zQ^{d-1-m}(Q^d)^{N-n}\\
    &=\sum_{n=0}^{N-1} T_z^{nd}A^{(d)}(Q^d)^{N-n}\\
    &=T_z^{(d\infty)}A^{(d)}\sum_{n=0}^{N-1}(Q^d)^{N-n} + \sum_{n=0}^{N-1}\Delta_z^{(n)}A^{(d)}(Q^d)^{N-n}
\end{split}
\end{equation}

where $\Delta_z^{(N)} = T_z^{dN}- T_z^{(d\infty)}$. Since $T_z^d$ is block diagonal, and each of the $d$ blocks $T_{zm}$ is irreducible and aperiodic, we have that $T^N_{zm} \to P_{zm}$ exponentially fast, i.e. $\|T^N_{zm} - P_{zm}\| = \|\Delta^{(N)}_{zm}\| < K_{zm} \tau_{zm}^N$, for some $K>0$ and $\tau<1$. Then, $\|\Delta^{(N)}_z\| \leq \sum_{m=0}^{d}\|\Delta^{(N)}_{zm}\| \leq K_z \tau_z^{N}$, where $\tau = \max_m \{\tau_{zm}\} < 1$. Together with \cref{prop:spectral_radius_bound}, this implies that, the last term in \cref{thm:main_dim:prop:chain} approaches zero with the same exponential speed, similarly to what happens in \cref{prop:vanishing_terms}.
Putting everything together, and considering that $\sum_{n=0}^{N-1}X^{n} \to (\mathds{1}-X)^{-1}$ $\forall X$ such that $\rho(X) < 1$, we have 
\begin{equation}
    A^{(d\infty)} =  T^{(d\infty)}A^{(d)}(\mathds{1}-Q^d)^{-1}
\end{equation}

At this point, by induction similarly to \cref{lemma:T_lim}, one can easily show that

\begin{equation}
    A^{(dN+m)} = T^{(dN)} A^{(m)} + A^{(dN)} Q^m
\end{equation}
which yields the proposition.
\end{proof}

In particular, the Cesàro average converges and we have the expression

\begin{equation}
A_z = \lim_{L\to\infty}\frac{1}{L}\sum_{l=1}^{L} A^{(l)} = \frac{1}{d}\sum_{m=0}^{d-1} T^{(d\infty)} A^{(m)} + A^{(d\infty)} Q^m.
\end{equation}

Recalling that the right eigenvector $v_z$ can be explicitly calculated when $\rho(T_z)=1$ (see \cref{cor:column_sum}), this allows to obtain the final form of $\Varr^\infty$ by ordinary matrix vector multiplication. As a special case, if all relevant blocks $T_z$ are aperiodic, then the limits of $T_z^L$ and $A_z^{(L)}$ converge, and we have $P_z = T_z^{(1 \infty)}$ and $A_z = A_z^{(1 \infty)} = P_z R_z (\mathds{1}-Q)^{-1}$.
Finally, the exponential upper bound in \cref{eq:cesaro_limit_app} and \cref{eq:simple_limit_app} is also obtained as a consequence of the preceding analysis. In particular if we denote by $T^\infty$ the matrix with $P_z$ in place of $T_z$, $A_z$ in place of $R_z$ and zero otherwise, we have

\begin{equation}
\begin{split}
    \left|\frac{1}{L}\sum_{l=0}^L \Varr^l - \Varr^\infty\right| &= \left|\left(\ell_\rho, \frac{1}{L}\sum_{l=0}^L T^l - T^\infty \ell_H\right)\right| \\
    &\leq \left\|\frac{1}{L}\sum_{l=0}^L T^l - T^\infty\right\| \sqrt{(\ell_\rho,\ell_\rho) (\ell_H,\ell_H)}.
\end{split}
\end{equation}

by Cauchy-Schwarz inequality. Since all blocks converge exponentially fast from the above discussion, the matrix norm is also exponentially decaying. Moreover, by construction $(\ell_A,\ell_A) \leq \sqrt{\sum_\kappa (\ell_A)^2_\kappa} \leq \sum_\kappa (\ell_A)_\kappa = \|A\|_2^2 \, \forall A\in\BH$, which concludes the proof. 
\end{proof}

While for an explicit calculation of $w_z$ one should in general rely on case-specific analyses, a general result can be derived for a subclass of channels especially useful in the context of quantum computing, namely single qubit noise.

\begin{proposition}[Single qubit noise]\label{cor:single_qubit_noise}
    Let $\rho, H \in \BH$ and let $\Phi_\theta$ a be layered quantum channel as in described in the main text. Assume moreover that the intermediate channel is of the form $\mathcal{E} = \mathcal{N}(W\rho W^\dagger)$, where $W$ is a unitary transformation and $\mathcal{N} = \mathcal{N}_1 \otimes ... \otimes \mathcal{N}_n$ is a composition of single qubit quantum channels. Then
    \begin{equation}
       \Varr^\infty = \sum_{z} \frac{(\ell_\rho)_z (\ell_H)_z}{d_z} + \frac{(\ell_\rho)_z (A \ell_H)_z}{d_z}.
    \end{equation}
\end{proposition}

\begin{proof}
Without loss of generality, we can consider the single qubit channels $\mathcal{N}_m$ to be in their normal form of \cref{lemma:normal_form}. In particular, this holds due to the invariance of the LTM with respect to changes of local bases (\cref{lemma:basis_independence_T}). In terms of the adjoint maps $\mathcal{N}^\dagger_m$, this condition reads $\mathcal{N}^\dagger_m (P_{j_m}) = t_{j_m} P_0 + \lambda_{j_m}P_j$ and can be used to compute $\Trb{\mathcal{N}^\dagger_m (P_{j_m})^2} = t_{j_m}^2 + \lambda_{j_m}^2 \leq 1$ by \cref{lemma:normal_form}. We now show that, when $P_j \in \BHz$ pertains to an irreducible component of spectral radius $\rho(T_z)=1$, then the inequality must be saturated. In particular, thanks to \cref{cor:column_sum} we know that $\sum_\kappa (T_z)_{\kappa, \lambda} = 1$. By \cref{eq:substochasticity} this implies
\begin{equation}\label{eq:unitary_noise}
\begin{split}
    \sum_\kappa (T_z)_{\kappa, \lambda} &= \frac{1}{d_\lambda} \sum_j \Trb{(W^\dagger \mathcal{N}^\dagger(P_j)W)^2} \deltat{j}{\lambda} \\&= \frac{1}{d_\lambda} \sum_j \Trb{\mathcal{N}^\dagger(P_j)^2} \deltat{j}{\lambda} \\
    &= \frac{1}{d_\lambda} \sum_j \prod_{m=1}^M \Trb{\mathcal{N}^\dagger_m(P_{j_m})^2} \deltat{j}{\lambda} = 1.
\end{split}
\end{equation}
Since each term in the product is upper-bounded by $1$, \cref{eq:unitary_noise} implies $\Trb{\mathcal{N}^\dagger_m (P_{j_m})^2} = t_{j_m}^2 + \lambda_{j_m}^2 = 1 \, \forall j_m$. This condition is only compatible with \cref{lemma:normal_form} if $t_{j_m}^2=0$ and $\lambda_{j_m}^2 = 1$. This result can now be used to show that, the adjoint of $T_z^\dagger$ of the LTM $T_z$, must also be column stochastic. Indeed, if we consider the expansion of $WP_jW^\dagger$ with respect to the normalized Pauli basis, we can get

\begin{equation}
\begin{split}
    \sum_\kappa (T_z)^\dagger_{\kappa, \lambda} &= \frac{1}{d_\lambda} \sum_j \Trb{\mathcal{N}(WP_jW^\dagger)^2} \deltat{j}{\lambda} \\
    &= \frac{1}{d_\lambda} \sum_j \Trb{\left(\sum_i \Trb{P_i WP_jW^\dagger} \mathcal{N}(P_i)\right)^2} \deltat{j}{\lambda}\\
    &= \frac{1}{d_\lambda} \sum_j \Trb{\left(\sum_i \Trb{P_iWP_jW^\dagger} \prod_m\lambda_{i_m} P_i\right)^2} \deltat{j}{\lambda} \\
    &= \frac{1}{d_\lambda} \sum_{i,j}\prod_m\lambda_{i_m}^2\Trb{P_iWP_jW^\dagger}^2\deltat{j}{\lambda}\\
    &= \frac{1}{d_\lambda} \sum_j \Trb{WP_j^2W^\dagger}\deltat{j}{\lambda} = 1.
\end{split}
\end{equation}

This allows to compute the right eigenvector $w_z$ of the leading eigenvalue of $T_z$. Using \cref{eq:detailed_balance}, we have indeed

\begin{equation}
    \sum_\lambda (T_z)_{\kappa, \lambda} d_\lambda =  \sum_\lambda (T_z)^\dagger_{\lambda,\kappa} d_\kappa = d_\kappa
\end{equation}
which implies $(w_z)_\lambda = d_\lambda/d_z$, where the normalization factor $d_z = \sum_{\lambda} d_\lambda$ is necessary to ensure $w_zv_z^t = P_z$ is indeed a projection, thus concluding the proof.
\end{proof}

As a consequence of this last result, we can compute the variance of generic unitary circuits. 

\begin{corollary}[Unitary circuits]\label{cor:deep_unitary_app}
For unitary circuits of type \cref{eq:circuit_channel}, we have

\begin{equation}
    \Varr^\infty = \sum_{z>0} \frac{(\ell_{\rho})_z (\ell_{H})_z}{d_z}
\end{equation}
\end{corollary}

\begin{proof}
    This form of the variance is a special case of \cref{cor:single_qubit_noise}, putting $\mathcal{N}(\rho) = \rho$, and noting that the absorption terms must vanish. In particular, this follows from \cref{eq:substochasticity}, observing that unitary channels saturate Kadison-Schwarz inequality, which combined with \cref{cor:column_sum} imply $Q=0$.
\end{proof}

On the opposite limit, if the noise map is strictly contractive in at least one direction in each $\BHz$, then the combination of noise and entanglement is strong enough to kill the variance in each of the absorbing subspaces. As a consequence, only the absorption term to $\BH_{0}$ remains, since no channel can be contractive there by trace preservation.

\begin{corollary}[Noise-induced concentration]\label{cor:deep_noisy_app} Let $\mathcal{E}$ be a quantum channel, and let $\{P_j\}$ denote the normalized Pauli basis. If $\|\mathcal{E}^\dagger (P_j)\|_2 < 1$ for some $j\in T_z$, $\forall z$, then
\begin{equation}
    \Varr^\infty = \frac{(A \ell_H)_0}{d}.
\end{equation}
In particular, if the channel is unital, $\Varr^\infty=0$.
\end{corollary}

\begin{proof}
    This form of the variance is a special case of \cref{eq:v_infty_general_app}, where all absorbing components vanish, and we have $T=Q$. In particular, this follows from \cref{eq:substochasticity}, observing that the above condition implies that for $\mathcal{E}^\dagger$ Kadison-Schwarz inequality is strict, which combined with \cref{cor:column_sum} imply $T_z=0\, \forall z>0$, leaving only $T_0$. Finally, the corollary follows from the normalization condition $\Trb{\rho}=1$ on $\rho$, which ensures $(\ell_\rho)_0=1/d$.
\end{proof}

\section{Variance lower bounds}\label{appendix:lower_bounds}

In this section we prove the general variance lower bounds outlined in \cref{sec:results}.

\subsection{Proof of \cref{thm:lower_bound}}

Here we employ \cref{prop:general_formula} to prove a general lower-bound on slowly entangling circuits. In particular, such result is based on the approximation $T_l\approx\mathds{1}$, which holds either for shallow circuits, i.e. $L\in O(\log n)$, or deeper circuits, but with weakly entangling intermediate channels. The discussion is based on the following result.

\begin{theorem}\label{thm:appendix:2}
    Let $\rho, H \in \BH$ and consider a sequence of quantum channels $\{\mathcal{E}_l\}_{l=1}^L$, and let $\{T_l\}_{l=1}^L$ be the respective LTMs. Finally let $K\subset \{0,1\}^M$ denote a subset of indices, and by $\alpha_l = \min_{\kappa\in K} (T_l)_{\kappa,\kappa}$. Then
    \begin{equation}\label{eq:lower_bound}
        \Varr^L \geq  \alpha^L (\ell_\rho,\ell_{\mathcal{K}(H)}),
    \end{equation}
    where $\mathcal{K}(\cdot) =\sum_{\kappa \in K} \sum_{j} \Trb{P_j \; \cdot} P_j\deltat{j}{\kappa}$ is a projector onto the space spanned by $K$ and $\alpha = \left(\prod_{l=0}^L \alpha_l\right)^{1/L}$ is the geometric mean of $\alpha_l$. 
\end{theorem}

\begin{proof}
    Consider a single circuit layer. Then, we can write
    \begin{equation}
    \begin{split}
        (\ell_\rho, T_l \ell_H) &= \sum_{\kappa, \lambda} \frac{(\ell_\rho)_\kappa T_{\kappa,\lambda} (\ell_H)_\lambda}{d_\kappa} \\
        &\geq \sum_{\kappa \in K} \frac{(\ell_\rho)_\kappa T_{\kappa,\kappa} (\ell_H)_\kappa}{d_\kappa} \\&\geq \alpha_l \sum_{\kappa\in K} \frac{(\ell_\rho)_\kappa (\ell_H)_\kappa}{d_\kappa} = \alpha_l (\ell_\rho, \ell_{{\mathcal{K}(H)}})
    \end{split}
    \end{equation}

where the inequality holds since all terms in the sum are non-negative by construction. The claim follows from repeated application of the latter.
\end{proof}

Despite its simplicity, \cref{thm:appendix:2} can be used to deduce general bounds on weakly entangling circuits, which are the foundation of small angle initialization strategies. In particular, we get the following Corollary.

\begin{corollary}\label{cor:small_angle:appendix}
    Let $\mathcal{H} = \bigotimes_{m=1}^M \mathcal{H}_m$, $d_m \in \Theta(1)$. If either of the conditions 
    \begin{enumerate}[(a)]
        \item $\alpha > 0$, $\alpha \in \Omega(1)$ and $L\in O(\log n)$,
        \item $\alpha = 1-f(n,L)$, $f \in O(\log n/L)$ and {$L \in \Omega(\log^{1+\epsilon}n)$} for some arbitrary $\epsilon > 0$
    \end{enumerate}
    is satisfied, then
    \begin{equation}
        \Varr^L \geq F(n) (\ell_\rho, \ell_{{\mathcal{K}(H)}}),
    \end{equation}
    where $F(n) \in \Omega(1/\text{poly}(n))$.    
\end{corollary}

\begin{proof}
    Exploiting \cref{thm:appendix:2}, it suffices to show that $F(n) = \alpha^L \in \Omega(1/\text{poly}(n))$. In the first case, this follows directly from the shallow nature of the circuit, and in particular $F(n) \in \alpha^{O(\log(n))} = \Omega(1/n^{-\log(\alpha)}) \subset \Omega(1/\text{poly}(n))$. For the second case instead, it is useful to consider $\log(F(n))$:
    \begin{equation}
    \begin{split}
        &-\log(F(n)) < -L\log\left(1-C\frac{\log(n)}{L}\right) \\&= C\log(n)\left(1+C\frac{\log(n)}{2L} + O\left(\frac{\log^2(n)}{L^2}\right)\right) \in O(\log(n))
    \end{split}
    \end{equation}
which in turn implies $F(n) \in e^{-O(\log(n))} = \Omega(1/n^C) \subset \Omega(1/\text{poly}(n))$.
\end{proof}

\subsection{Small-angle initializations lower bounds}

In order to prove the general lower bounds on small angle initializations provided in the main text, it is useful to start from \cref{prop:noise_small_angle_equiv}, as it is the fundamental building block in this type of proofs. We recall it for convenience.

\begin{proposition}\label{prop:equivalence:appendix}
    Let $\{E_\phi\}_\phi$ be an ensemble such that $\mathcal{E}(\rho) = \mathbb{E}_\phi \{ E_\phi \rho E_\phi^\dagger\}$ is a quantum channel. Further, denote by $T_\phi$ the transfer matrix associated to each $\mathcal{E}_\phi^\dagger(\cdot) = E^\dagger_\phi \cdot E_\phi$, and by $T$ the transfer matrix of $\mathcal{E}^\dagger$. Then we have
    \begin{equation}
\mathbb{E}_\phi \{T_\phi\} \geq T,
    \end{equation}
with equality holding if and only if $\mathcal{E}$ is unitary.
\end{proposition}

\begin{proof}
    Let $A,B$ be arbitrary bounded operators, and consider
    \begin{equation}
    \begin{split}
        \Trb{\mathcal{E}(A)B}^2 &= \Trb{A\mathcal{E}^\dagger(B)}^2 \\
        &= \mathbb{E}_\phi\left\{\Trb{A\mathcal{E}^\dagger_{\phi}(B)}\right\}^2 \\
        &\leq \mathbb{E}_\phi\left\{\Trb{A\mathcal{E}^\dagger_{\phi}(B)}^2\right\} \;\; \forall A,B
    \end{split}
    \end{equation}
    which follows from the observation that $f(\phi) = \Trb{A\mathcal{E}^\dagger_{\phi}(B)} \in \mathbb{R}$, and so $\Var{\phi}\{f\} \geq 0$. Applying this to the entries of $T$ and $T_\phi$ gives the general inequality. Finally, the equality follows from $\Var{\phi}\{f\}=0$, which is means that $f(\phi)=\Trb{AK^\dagger B K}$ is a constant, where $K=E_\phi$ $\forall \phi$. Hence, since $\mathcal{E}(\cdot) = K \cdot K^\dagger$ is CPTP, it must also be unitary. 
\end{proof}

In order to translate this rather abstract formulation into a practical recipe, we need to identify the conditions that allow to treat the contribution of given an ensemble $\{\mathcal{E}_\phi\}$ of parameterized intermediate channels to the variance in terms of the mean LTM $\mathbb{E}_\phi\left\{T_\phi \right\}$. In particular, it is easily verified that, if $\phi$ is sampled independently of the other parameters, then

\begin{equation}
\begin{split}
     \Varr^L &= \mathbb{E}_\phi \mathbb{E}_\theta \left\{\Trb{\Phi_{\theta,\phi} (\rho) H}^2\right\} \\
     &= \mathbb{E}_\phi \left\{(\ell_\rho, T_\phi \ell_H)\right\} =  (\ell_\rho, \mathbb{E}_\phi \{T_\phi\} \ell_H).
\end{split}
\end{equation}

Combining this observation with \cref{cor:small_angle:appendix}, we can get the QResNet lower bound of \cref{prop:qresnet}.

\begin{proposition}[QResNet]
    Let $\mathcal{E}_\phi(\cdot) = e^{i\phi G} \cdot e^{-i\phi G}$ be a unitary entangling gate, and let $\mu, \sigma^2$ be the mean and variance of the initialization distribution $\mathcal{P}$ of $\phi$. Then, if $\mu=0$, ${\sigma^2 \in O(\log n/\|G\|_2^2L)}$, and $L\in\Omega(\log^{1+\epsilon} n)$
    \begin{equation}
        \Varr^L \geq F(n) (\ell_\rho, \ell_H)
    \end{equation}
    where $F(n) \in \Omega(1/\text{poly}(n))$.
\end{proposition}

\begin{proof}
    Let $T$ be the locality transfer matrix of $\mathcal{E} = \mathbb{E}_\phi\{\mathcal{E}_\phi\}$, and consider the diagonal element $T_{\kappa,\kappa}$. Then, by definiton, we have

    \begin{equation}\label{eq:qresnet_appendix:1}
    \begin{split}
                T_{\kappa,\kappa} &= \frac{1}{d_\kappa} \sum_{i, j} \Trb{\mathbb{E}_\phi\{P_i\mathcal{E}^\dagger_\phi(P_j)\}}^2 \Tilde{\delta}_{i, \kappa} \Tilde{\delta}_{j, \kappa} \\
                &\geq \frac{1}{d_\kappa} \sum_{i} \Trb{\mathbb{E}_\phi\{P_i\mathcal{E}^\dagger_\phi(P_i)\}}^2 \Tilde{\delta}_{i, \kappa} \\
                &= \frac{1}{d_\kappa} \sum_{i} \Trb{\mathbb{E}_\phi\{P_ie^{-i\phi G}P_i e^{i\phi G}\}}^2 \Tilde{\delta}_{i, \kappa}
    \end{split}
    \end{equation}
Since $\sigma^2\to 0$ as $n\to \infty$, to find the asymptotic behaviour of the diagonal elements of $T$ we can expand $e^{i\phi G}$ around $\mu$, and obtain ${e^{i\phi G} = \mathds{1} + i\phi G - \phi^2 G^2/2 + O(\phi^3\|G\|_2^3)}$. Substituting this into \cref{eq:qresnet_appendix:1}, we get 
\begin{equation}
    \begin{split}
        T_{\kappa,\kappa} &\approx \frac{1}{d_\kappa} \sum_{i} \text{Tr}\Big[\mathbb{E}_\phi\{P_i(\mathds{1} + i\phi G - \phi^2 G^2/2)P_i (\mathds{1} +\\
        &+i\phi G - \phi^2 G^2/2)\}\Big]^2 \Tilde{\delta}_{i, \kappa}\\
        &= \frac{1}{d_\kappa} \sum_{i} \left(1-\mathbb{E}_\phi\{\phi^2\}\left(\Trb{G^2P_i^2} - \Trb{P_iGP_iG}\right)\right)^2\Tilde{\delta}_{i, \kappa}\\
        &\geq 1- 4\|G\|^2_2 \sigma^2 \;\; \forall \kappa\in\{0,1\}^M
    \end{split}
\end{equation}
Exploiting \cref{cor:small_angle:appendix}, we get the proposition by showing $T_{\kappa,\kappa} \geq 1 - f(n,L)$, where $f(n,L) \in O(\log n/L)$. In particular this follows directly from the scaling of $\sigma^2$. The same proof ensures absence of concentration on a unitary QResNet, provided that $\mathcal{P}$ is chosen as an initialization probability by \cref{prop:equivalence:appendix}.
\end{proof}

\section{NIBP-free landscapes under unital noise}\label{sec:unital_noise_examples}
To illustrate the emergence of non-trivial behavior in the loss variance $\Varr$ under unital noise channels, we focus on two representative examples. Both involve small quantum circuits composed of $n=2$ qubits and correlated noise. In these cases, the presence of non-unitary layers is either irrelevant or beneficial to the overall magnitude of the variance.

\begin{figure}
    \centering
    \includegraphics[width=\linewidth]{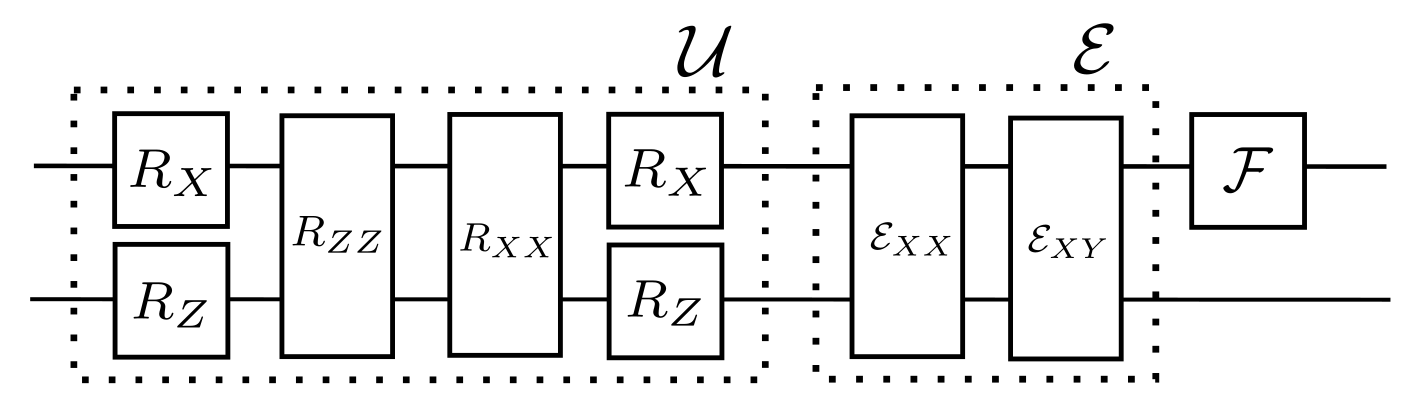}
    \caption{Single layer structure used in the numerical examples. The circuit acts on a two-qubit register and is divided into three blocks, each representing either unitary gates ($\mathcal{U}$) or noise channels ($\mathcal{E}$ and $\mathcal{F}$).}
    \label{fig:circ_example}
\end{figure}

We begin with the circuit depicted in \cref{fig:circ_example}, which features a unitary layer $\mathcal{U}$ as shown, combined with a fixed noise channel $\mathcal{E}_l = \mathcal{E} = \mathcal{E}_{XX} \circ \mathcal{E}_{XY} \; \forall l$. Specifically, $\mathcal{E}$ is defined as the composition of two dephasing-like channels on the $XX$ and $XY$ bases respectively, namely $\mathcal{E}_{XX}(\rho) = (1-p)\rho + p XX \rho XX $ and $\mathcal{E}_{XY}(\rho) = (1-p)\rho + p XY\rho XY$.

By analyzing the generators of the parameterized gates $R_X$ and $R_{ZZ}$ (i.e., $X$ and $ZZ$), one can verify that they commute with both the other rotations appearing in $\mathcal{U}$ (i.e., $R_Z$ and $R_{XX}$) and with the noise channel for any choice of parameters. Consequently, their action can be moved to the end of the circuit, leaving a mixed channel at the beginning. Notably, the state $\rho = \ket{00}\bra{00} + \ket{11}\bra{11}$ is a fixed point of this mixed channel, as it is singularly as fixed point of each of the remainig operations. As a result, by choosing $\rho$ as the initial state, the circuit is effectively equivalent to its purely unitary counterpart and thus free of NIBP. This is confirmed by an analytical computation of the variance via \cref{eq:v_infty_general}, which yields a non-vanishing result that is independent of the number of layers. In particular, \cref{fig:example1} numerically corroborates the prediction of $\Varr^L = 1$ for $H = ZZ + XI + YZ$. 

Furthermore, by substituiting $\mathcal{U}$ with $\mathcal{U}'$, obtained by replacing the first $RX$ gate with an $RY$ gate, \cref{fig:example1} shows how the NIBP phenomenon reappears, highlighting again the crucial role in the interplay between noise and circuit structure.

\begin{figure}
    \centering
    \includegraphics[width=\linewidth]{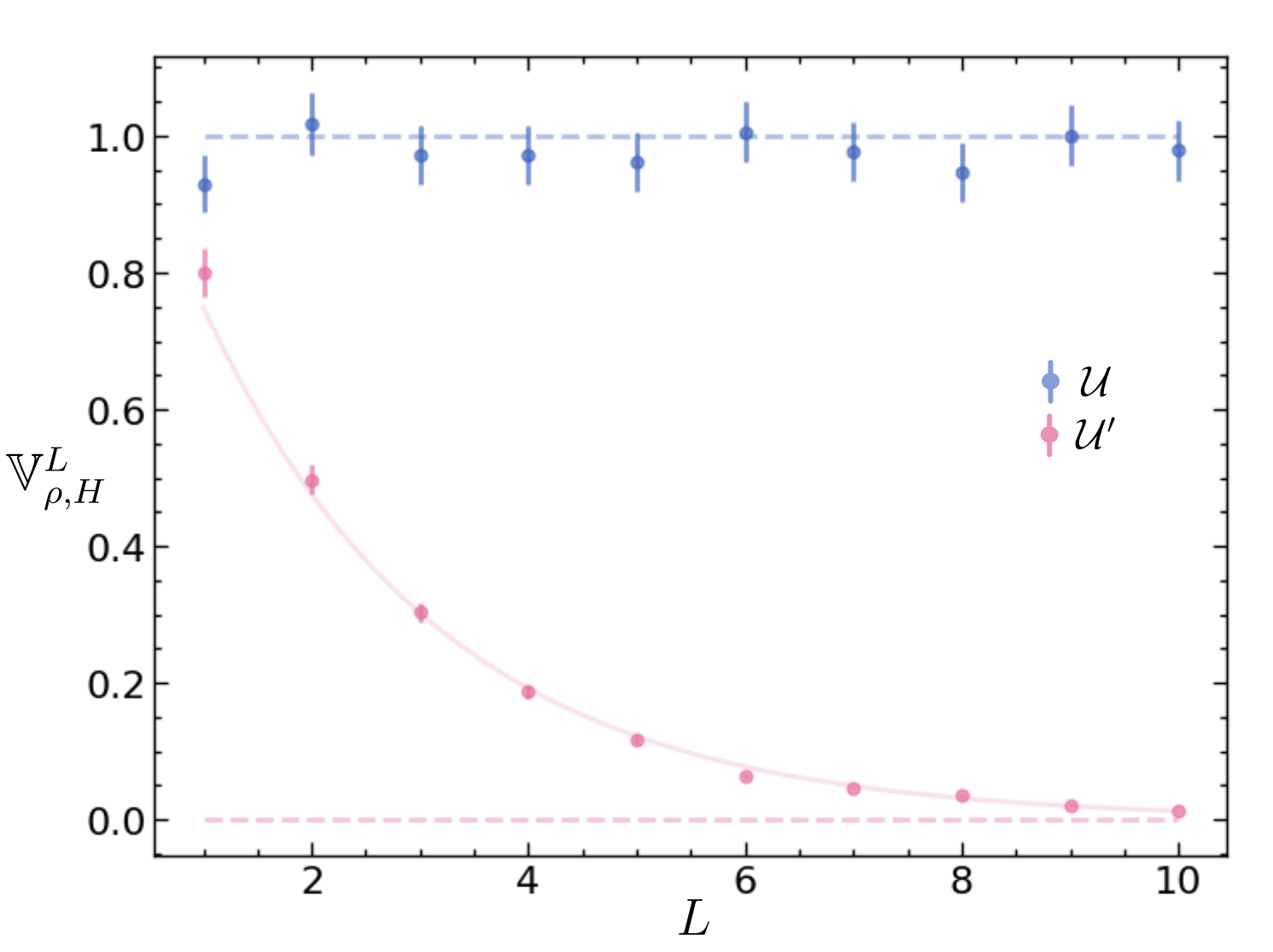}
    \caption{Example of NIBP-free unital, non-unitary channel. Dotted lines represent the theoretical deep-circuit limiting values for the variance in the case of layers defined by $\mathcal{U}\circ\mathcal{E}$ (blue) and $\mathcal{U}'\circ\mathcal{E}$ (pink), where $\mathcal{U}'$ is obtained from $\mathcal{U}$ in \cref{fig:circ_example} by substituiting the first $RX$ gate with an $RY$ gate. The solid pink line is an exponential fit of the convergence to the limiting value.}
    \label{fig:example1}
\end{figure}

As a second example, we change initial state, setting it to $\rho = \mathds{1}\otimes \ket{0}\bra{0}$, while keeping $H$. By virtue of this choice, the variance of the unitary circuit is now identically zero, as shown in \cref{fig:example2}. Here, we consider a modified noise channel $\mathcal{E}' = \mathcal{E} \circ \mathcal{F}$, where $\mathcal{F}$ is a single-qubit channel applied to the first qubit, defined as $\mathcal{F}(\rho) = (1-p)\rho + p\mathcal{A}(\rho)$, and
\begin{equation}
    \begin{split}
    \mathcal{A}(\rho) &= 
 \frac{q}{\sqrt{2}}\Trb{X\rho} X 
 +\frac{q}{\sqrt{2}}\Trb{X\rho} Y\\ 
&+(1-q)\Trb{Y\rho} Y \\
&+\frac{q}{\sqrt{2}}\Trb{Z\rho} Z
+\frac{q}{\sqrt{2}}\Trb{Z\rho} Y
    \end{split}
\end{equation}
By computing its Choi matrix, one can confirm that $\mathcal{F}$ is completely positive and trace-preserving (CPTP), and therefore constitutes a valid quantum channel whenever $0\leq p \leq 1$ and $0\leq q \leq 1/2$. Although $\mathcal{E}'$ is still not strictly contractive, it is \emph{contractive enough} to satisfy \cref{cor:deep_noisy}, and thus provably induces the NIBP effect, as corroborated numerically in \cref{fig:example2}. Interestingly however, the behavior of $\Varr$ under this channel is not monotonic with respect to the circuit depth: it initially increases before eventually decaying to zero in the deep-circuit limit, thus temporarily enhancing the variance.
This observation motivates the consideration of a non-uniform noise profile, where $\mathcal{E}_l \neq \mathcal{E}_{l'}$ for some $l, l'$, defined as: 
\begin{equation}
    \mathcal{E}_l'' = \begin{cases}
        \mathcal{E}' \;\; &\text{if } l \leq const  \\
        \mathcal{E} \;\; &\text{otherwise}\\
    \end{cases}
\end{equation}
In such a setup, the circuit benefits from the absence of concentration observed in the first example, as well as the transient variance enhancement from the second, converging to non-zero value dependent on the constant rather than the number of layers. This hybrid effect is clearly demonstrated in \cref{fig:example2}. The mechanism at play here is at the core of the \emph{absorption} effects shown in the main text.

\begin{figure}
    \centering
    \includegraphics[width=\linewidth]{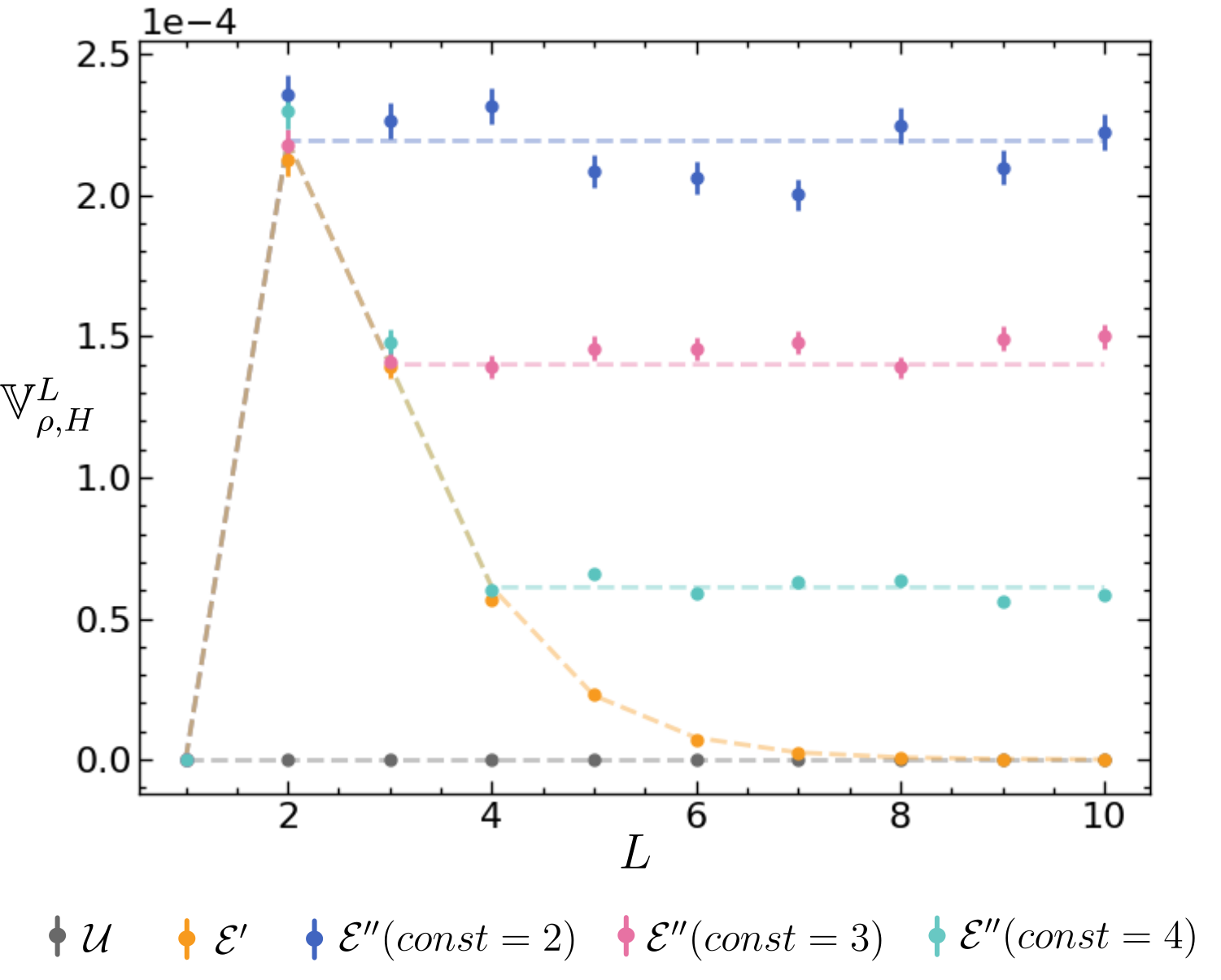}
    \caption{Variance improvment with unital channels. Dotted lines represent the theoretical prediction given by \cref{eq:general_formula}, while points indicate the numerically estimated circuit variances. }
    \label{fig:example2}
\end{figure}

Finally, we emphasize that, despite the limited system size considered ($n=2$ qubits), these examples are still insightful in illustrating the absence of NIBP. This is because the key quantity of interest is the scaling with respect to the circuit depth $L$, rather than the number of qubits $n$.

\section{Non-unital noise and entanglement}\label{sec:methods:noise_ent}

In this Section we explicitely derive the absorption terms pertaining to the trivial component $\mathcal{B}_0$, in the presence of strictly contractive, yet non-unital noise, first in general, and then on a specific system used for the numerical calculations.

\subsection{Numerics}

\begin{figure}
    \includegraphics[width=\linewidth]{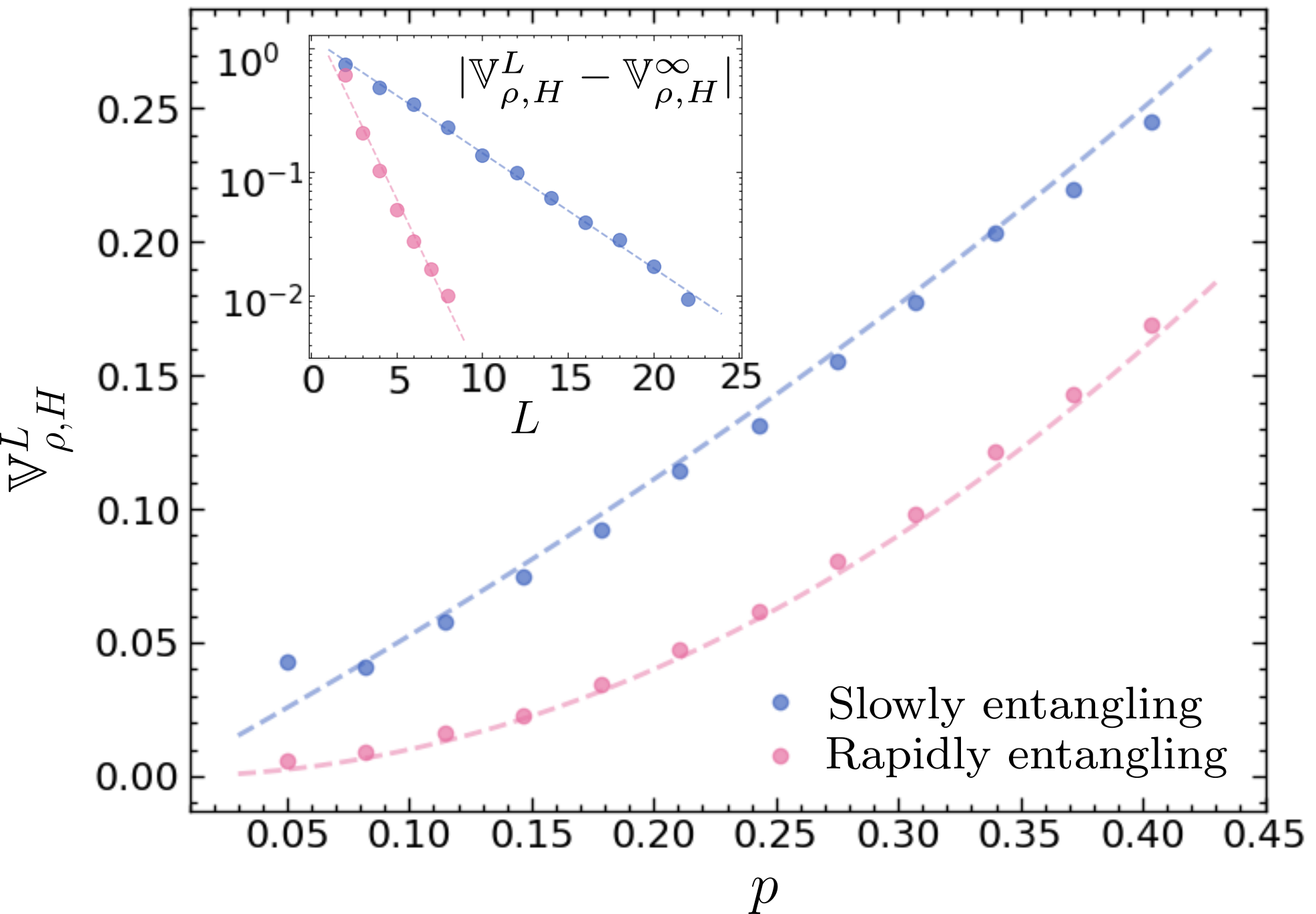}
    \caption{Scaling of $\Varr^L$ as a function of the noise strength and the entangling power of the intermediate channel. The main figure illustrates the scaling of $\Varr^\infty$ with noise strength $p$ for both rapidly entangling (pink) and slowly entangling (light blue) channels, using $L = 8$ and $L = 20$, respectively. The dotted lines represent the theoretical predictions of \cref{eq:quadratic_dependence} and \cref{eq:linear_dependence} The inset verifies the exponential convergence of $\Varr^L$ to $\Varr^\infty$ at $p = 0.1$, justifying the chosen number of layers. The dotted lines represent an exponential fit to the numerical data. All plots are obtained using a $n=10$ qubit system.}
    \label{fig:numerical}
\end{figure}

Exploiting the knowledge about the structure of the absorption matrix $A$, derived from \cref{thm:main}, it is possible to study the scaling of the variance $\Varr^\infty$ as a function of the noise strength and entangling power of the unitary circuit. To do so, let us consider a non-unital map of the form
\begin{equation}\label{eq:noise_model}
    \mathcal{E}_c(\rho) = (1-p)\mathcal{E}(\rho) + p \Tilde{\rho},
\end{equation}
where $\Tilde{\rho}\neq\mathds{1}/d$ is an arbitrary quantum state, $\mathcal{E}$ is a unitary channel representing the entangling operation, and $p$ is the error probability associated to $\mathcal{E}_c$. Intuitively, we can think of the resulting channel $\Phi_\theta$ as the repetition of $L$ layers, each made up of the composition of $\Tilde{\Phi}(\rho) = (1-p)\rho + p\Tilde{\rho}$ and of $\mathcal{E}\circ\mathcal{U}_{\theta_l}$. Clearly $\Tilde{\rho}$ is the unique fixed-point of $\Tilde{\Phi}$, while since the local unitaries are assumed to form 2-designs, the only fixed point for the unital part, valid for all parameters, is the maximally mixed state $\mathds{1}/d$. This causes the emergence of competing effects, which are the ultimate origin of the variance in such models. However, depending on the relative strength of the two effects, the behaviour of $\Varr^\infty$ as a function of $p$ may vary greatly. \cref{thm:main} allows to quantify the variance scaling in the limit of a rapidly entangling and slowly entangling channel, i.e. $\beta \to \infty$ and $\beta \to 0$ respectively. In particular, for both limits, the LTM approaches a projection, namely the dominant eigenprojection in the former and the identity in the latter. In these cases, $\Varr^\infty$ is given by the following Lemma.

\begin{lemma}\label{lemma:noise_model}
Let $\mathcal{E}_c$ be a quantum channel of type \cref{eq:noise_model}, with $0< p \leq 1$ and let $T$ be the transition matrix of $\mathcal{E}^\dagger$. Then
    \begin{equation}
        \Varr^\infty = p^2 (\ell_{\Tilde{\rho}}, (\mathds{1}-(1-p)^2T)^{-1} \ell_H).
    \end{equation}
In particular, if $T$ is a projection, then
    \begin{equation}\label{eq:noise_scaling_projector}
        \Varr^\infty = \left(\frac{p}{2-p}-p^2\right)(\ell_{\Tilde{\rho}},T\ell_H) + p^2(\ell_{\Tilde{\rho}},\ell_H).
    \end{equation}
\end{lemma}

The first thing to notice is that the dependence on the initial state of $\Varr^\infty$ is completely lost: the component associated to it decays exponentially fast in the number of layers when $p\in\Theta(1)$, and as a consequence, vanishes in the limit. The remaining terms, instead, pertain to the fixed point of the noise channel $\Tilde{\rho}$, and therefore still appear. In particular, the last term pertains to the very last layer, while the first collects the contribution of all preceding layers. Clearly, since the decay of these contributions is exponential, only layers where $L-l \lesssim -2\log p$ will contribute sensibly to the variance.

\cref{lemma:noise_model} allows to compute the scaling as a function of $p$ in the two opposite limits. Starting from the slowly entangling case, i.e. $T \approx \mathds{1}$, the scaling is approximately \emph{linear} in $p$. More precisely, we have
\begin{equation}\label{eq:linear_dependence}
    \Varr^\infty = \frac{p}{2-p}(\ell_{\Tilde{\rho}}, \ell_H).
\end{equation}
The result shows that first term in \cref{eq:noise_scaling_projector} dominates, suggesting that $\Varr^\infty$ emerges from the contribution of the last $O(-\log p)$ layers. In this sense, for fixed noise rates, only the last portions of the channels are relevant for VQAs \cite{Mele24}. Conversely, in the opposite limit, the dependence on $p$ is more complex, as it now depends on the irreducible components $T_z$. For simplicity, if we consider the case of a highly expressive ansatz, we may take $T$ to have only one irreducible component $T_1$. In this case, the last term in \cref{eq:noise_scaling_projector} dominates, and we get a \emph{quadratic} dependence on $p$ up to an exponentially vanishing correction, namely 
\begin{equation}\label{eq:quadratic_dependence}
    \Varr^\infty = p^2(\ell_{\Tilde{\rho}}, \ell_H) + O\left(4^{-n}\right).
\end{equation}
This worsens the concentration, suggesting that now only the very last layer is able to produce a sizable effect, hence giving an effectively constant depth circuit.

We provide numerical evidence for the application discussed here, derived by the formalism introduced in this work.
We utilize Pennylane~\cite{Bergholm22} to construct and optimize the PQC described in the following.

Regarding the initial state, we use $\rho=(\ket{0}\bra{0})^{\otimes n}$ for simplicity.
As for the channel, we consider maps $\mathcal{E}_c$ of the family $\mathcal{E}_c = \mathcal{N}\circ\mathcal{E}$, $\mathcal{E}_c(\rho) = (1-p)\mathcal{E}(\rho) + p \Tilde{\rho}$. In particular, $\mathcal{E}$ is a unitary, entangling channel  depicted in \cref{fig:ents} and $p\in (0,1]$ represents the noise strength of the noise map $\mathcal{N}(\rho) = (1-p)\rho + p\Tilde{\rho}$ with fixed point $\Tilde{\rho}$. 
Specifically, we fix $\Tilde{\rho}$ to be a highly entangled, pure state, i.e. the GHZ state $\Tilde{\rho} = (\ket{0}^{\otimes n} + \ket{1}^{\otimes n})(\bra{0}^{\otimes n} + \bra{1}^{\otimes n})/2$.

\begin{figure}
  \includegraphics[width=\linewidth]{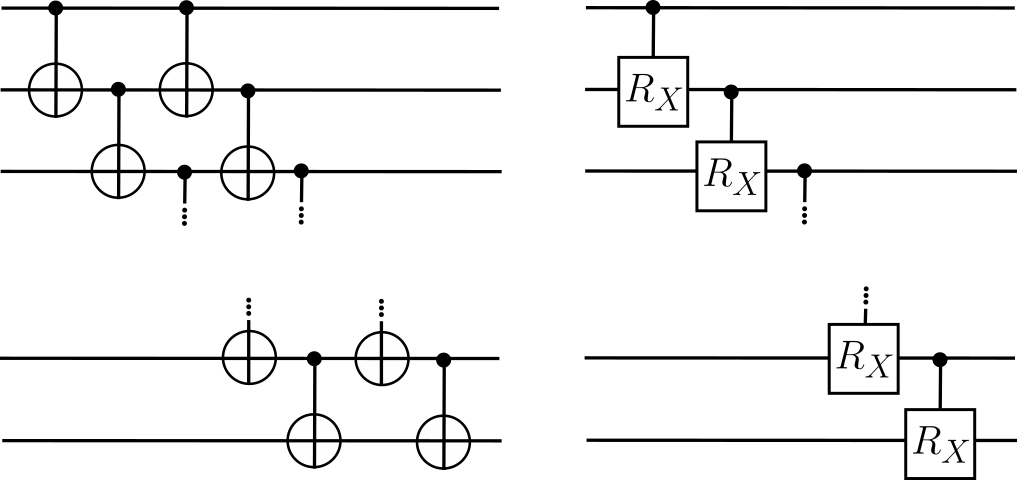}  
  \caption{Entangling unitaries used in the examples. On the left, the rapidly entangling configuration is composed of a double cascade of CNOT gates, while on the right, the slowly entangling one is composed of a single cascade of controlled $RX$ gates, where $RX(\theta) = e^{i\theta X/2}$ and $X$ is the Pauli $X$ gate. In particular, we fixed $\theta= \pi/20$.} \label{fig:ents}
\end{figure}

Concerning $H$, we fix $H = h\sum_{k=1}^{n} 2^{n/2}\,Z_{k}\otimes Z_{k+1}$, as it represents the simplest observable involving all $n$ qubits while having a non-vanishing normalization factor $(\ell_{\Tilde{\rho}}, \ell_H)$. In particular, as shown in \cref{app:sec:numerical_example}, if we fix $h=9/n$, we have $(\ell_{\Tilde{\rho}}, \ell_H) = 1$, which makes the scaling of $\Varr^\infty$ especially easy to check.
Finally, the entangling part is chosen according to \cref{fig:ents}. While the number of layers $L$ needed to reach convergence to $\Varr^\infty$ is logarithmic, as shown in \cref{thm:main} and numerically assessed in the inset of \cref{fig:numerical}, the mixing speed varies depending on the entangling part. For this reason, $L=8$ layers are sufficient in the rapidly entangling case, but $L=20$ are necessary in the slowly entangling case. Finally, all simulations are performed using $n=10$ qubits.
The results are shown in \cref{fig:numerical}. The plots show both the quadratic and linear scaling with $p$ predicted by our model, as well as an exponential decay in the difference $|\Varr^L-\Varr^\infty|$ for fixed noise rates. A slight deviation from the predicted scaling is observed in the slowly entangling setting for $p \approx 0$. This effect can be explained by the finite amount of layers used: the slow speed of convergence imposed by the condition $T\approx \mathds{1}$ prevents the variance to reach the asymptotic limit $\Varr^\infty$, while the action of the noise is still too weak to erase the contribution of the first layers, hence deviating from the predicted behaviour. This same phenomenon is not observed in the rapidly entangling case, as entanglement production already exponentially suppresses those contributions.

\subsection{Calculation of the absorption term}
Explicit computation of the absorption matrix $A = R(\mathds{1}-Q)^{-1}$ is a non-trivial task, which should be tackled on a case-by-case basis. Indeed, in the non-unital case, this term describes the complex phenomenon arising from the interaction of two competing effects, which drive the system towards different states. Analytical summation of $A$ is however can be feasible and still give insight on the interaction of the two. Indeed, using a simplified model, we can perform this calculation and still be able to appreciate the different effects that rapidly entangling and slowly entangling circuits have on a fixed, non-unital noise as a function of its strength. In particular, here we prove \cref{lemma:noise_model}, which we repeat for convenience.

\begin{lemma}
Let $\mathcal{E}_c$ be a quantum channel of type \cref{eq:noise_model}, with $0< p \leq 1$ and let $T$ be the transition matrix of $\mathcal{E}^\dagger$. Then
    \begin{equation}
        \Varr^\infty = p^2 (\ell_{\Tilde{\rho}}, (\mathds{1}-(1-p)^2T)^{-1} \ell_H).
    \end{equation}
In particular, if $T$ is a projection, then
    \begin{equation}
        \Varr^\infty = \left(\frac{p}{2-p}-p^2\right)(\ell_{\Tilde{\rho}},T\ell_H) + p^2(\ell_{\Tilde{\rho}},\ell_H).
    \end{equation}
\end{lemma}

\begin{proof}
The first result follows directly from \cref{thm:main_appendix}, in particular \cref{cor:deep_noisy_app}, by computation of $A_0$. In particular, we can explicitly compute $R_0$ element wise as

\begin{equation}
\begin{split}
       (R_0)_{0,\kappa} &= \frac{1}{d_\kappa} \sum_{j} \Trb{\frac{\mathds{1}}{\sqrt{d}}\mathcal{E}_c^\dagger(P_j)}^2 \Tilde{\delta}_{j, \kappa}\\
       &= \frac{1}{d_\kappa} \sum_{j} \Trb{\mathcal{E}_c\left(\frac{\mathds{1}}{\sqrt{d}}\right)P_j}^2 \Tilde{\delta}_{j, \kappa}\\
       &= \frac{1}{d_\kappa} \sum_{j} \Bigg((1-p)\Trb{\Phi\left(\frac{\mathds{1}}{\sqrt{d}}\right)P_j} \\
       &+ p\sqrt{d}\Trb{{\Tilde{\rho}}P_j} \Bigg)^2\Tilde{\delta}_{j, \kappa}\\
       &= dp^2 \frac{1}{d_\kappa} \sum_{j} \Trb{{\Tilde{\rho}}P_j}^2\Tilde{\delta}_{j, \kappa} = d p^2  (\ell_{\Tilde{\rho}})_\kappa
\end{split}
\end{equation}
by unitality of $\Phi$. By an analogous calculation, it can be shown that $Q=(1-p)^2T$, and by trace preservation $(P_0)_{0,0}=1$.
With these elements, we can compute $A_0 = P_0R_0(1-Q)^{-1}$, and we get
\begin{equation}
    (A_0)_{0,\lambda} = dp^2 (\ell_{\Tilde{\rho}})_\kappa (1-(1-p)^2T)^{-1}_{\kappa,\lambda}.
\end{equation}
In particular, since for any initial state $\rho$, we have $(\ell_\rho)_0 = 1/d$ by normalization, we get the final result
\begin{equation}
    (\ell_\rho, A_0 \ell_H) = p^2 (\ell_{\Tilde{\rho}}, (\mathds{1}-(1-p)^2T)^{-1} \ell_H).
\end{equation}
Using this, we can explicitly compute the right-hand side in the simplified setting where $T$ is a projection. In that case in particular, we have that 
\begin{equation}
\begin{split}
     &(\mathds{1}-(1-p)^2T)^{-1}\\
     = &(\mathds{1}-T) + (1-(1-p)^2)^{-1}T\\
     = &(\mathds{1}-T) + \frac{1}{p(2-p)}T \\
     = &\mathds{1} +
      \left(\frac{1}{p(2-p)} -1\right)T\\
\end{split}
\end{equation}
which gives the result.
\end{proof}

\subsection{Choice of the system for the numerical example}\label{app:sec:numerical_example}

In this section we give an explicit construction of the noise channels used in the numerical example shown in the main text, as well as the choice and normalization procedure of the observable $H$ used. 
In particular, the analysis focuses on the case of unitary entangling map $\Phi$ whose transfer matrix is well approximated by a projector, and a highly entangled, pure fixed point, namely the GHZ state $\Tilde{\rho} = (\ket{0}^{\otimes n} + \ket{1}^{\otimes n})(\bra{0}^{\otimes n} + \bra{1}^{\otimes n})/2$.
Thanks to \cref{lemma:noise_model}, we know that the main contribution to $\Varr^\infty$ in this setting comes from the term $(\ell_{\Tilde{\rho}}, \ell_H)$, so to numerically assess the results, it is useful to choose the observable $H$ in order to normalize this factor. In particular, given the qubit structure, we can express the $\Tilde{\rho}$ in terms of the \emph{normalized} Pauli basis $\{\mathds{1}, X,Y,Z\}^{\otimes n}$, which allows to easily compute the locality vectors. Using the spectral decomposition of the Pauli matrices, it is easy to see that 

\begin{equation}
\begin{split}
     &\ket{0}\bra{0} = \frac{\mathds{1} + Z}{\sqrt{2}}, \;\; \ket{0}\bra{1} = \frac{X + iY}{\sqrt{2}}, \\
     &\ket{1}\bra{1} = \frac{\mathds{1} - Z}{\sqrt{2}}, \;\; \ket{1}\bra{0} = \frac{X - iY}{\sqrt{2}}.
\end{split}
\end{equation}

Using this decomposition, we can find a formula for $\Tilde{\rho}$ exploiting a generalization of the binomial theorem.

\begin{theorem}\label{thm:permutation:appendix}
    Let $A, B \in \mathbb{M}_d(\mathbb{C})$ be square matrices, and let $S$ be the set of all permutations of $n$ elements. Then
    \begin{equation}
        (A + \omega B)^{\otimes n} = \sum_{j=0}^{n} \omega^j \sum_{\sigma \in S} \sigma(A^{\otimes (n-j)}\otimes B^{\otimes j})
    \end{equation}
    where $\omega \in \mathbb{C}$ and the permutation $\sigma$ is applied to the qubit ordering.
\end{theorem}
In particular, the following corollary will be the most useful in performing the computations
\begin{corollary}
    Let $A, B \in \mathbb{M}_d(\mathbb{C})$ be square matrices, and let $S$ be the set of all permutations of $n$ elements. Then
    \begin{equation}
        \frac{(A + \omega B)^{\otimes n} + (A-\omega B)^{\otimes n}}{2} =\sum_{j=0}^{\lfloor n/2 \rfloor} \omega^{2j} \sum_{\sigma \in S} \sigma(A^{\otimes (n-2j)}\otimes B^{\otimes 2j})
    \end{equation}
    where $\omega \in \mathbb{C}$ and the permutation $\sigma$ is applied to the qubit ordering.
\end{corollary}
\begin{proof}
    This corollary follow directly from \cref{thm:permutation:appendix}, and noticing that all even-indexed terms in both $(A + \omega B)^{\otimes n}$ and $(A -\omega B)^{\otimes n}$ are equal, while odd-numbered terms have opposite sign and therefore cancel out.
\end{proof}

Applying the Corollary to the appropriate pairs of projectors, we can get the final expression

\begin{equation}\label{eq:GHZ_rho:appendix}
\begin{split}
    \Tilde{\rho} = \frac{1}{2^{n/2}} \sum_{j=0}^{\lfloor n/2 \rfloor} \sum_{\sigma \in S} \sigma(\mathds{1}^{\otimes (n-2j)}\otimes Z^{\otimes 2j})\\ + (-1)^j \sigma(X^{\otimes (n-2j)}\otimes Y^{\otimes 2j})
\end{split}
\end{equation}

As it is clear from \cref{eq:GHZ_rho:appendix}, the fixed point of the channel has a non-vanishing component only on Pauli strings that are either non-trivial on all qubits, or non-trivial in \emph{only} in an \emph{even} number of qubits. Then it follows that the simplest observable involving all qubits and with non-vanishing variance is of form $H = h\sum_{k=1}^{n} 2^{n/2}\,Z_{k}\otimes Z_{k+1}$, where the factor $2^{n/2}$ accounts for normalization of $Z$, and cancels out with the corresponding factor in \cref{eq:GHZ_rho:appendix} in the calculations of $(\ell_{\Tilde{\rho}}, \ell_H)$. Finally, since each term in the sum is orthogonal, it gives an independent contribution of $1/(d^2-1)^2 = 1/9$, consequently by choosing $h = 9/n$ we have that $(\ell_{\Tilde{\rho}}, \ell_H) = 1$ is normalized.

\section{Example of a circuit with non-convergent variance}\label{app:non_conv_var}

As \cref{thm:main} suggests, the generic circuit of \cref{eq:circuit_channel} need not have a well-defined value for the deep circuit limit of its variance, namely the limit $\lim_{L\to\infty} \Varr^L$ no need to exists. As discussed in \cref{appendix:proof_main}, this property is related to the presence of \emph{cycles} in $T$, i.e. the presence of periodic irreducible blocks $T_z$ with period $p>1$. As a specific example, consider the circuit depicted in \cref{fig:non_conv_variance}. The simple structure of this circuit allows to explicitly compute $\Varr^L$ as function of $L$. Assuming $\Trb{H}=0$, we have
\begin{equation}
    \Varr^L = \begin{cases}
        \frac{(\|\rho\|_2^2 -1/2)\|H\|_2^2}{3} \;\; &\text{if } L \text{ is even}  \\
        0 \;\; &\text{if } L \text{ is odd}  \\
    \end{cases}
\end{equation}
from which it is clear that the deep circuit limit does not converge. Moreover, thanks to the contained dimension of the system, it is possible to compute and represent $T$:
\begin{equation}
T = 
\begin{pmatrix}
1 & 0& 0& 0\\
0 & 0 & 1 & 0\\
0 & 1 & 0 & 0\\
0 & 0 & 0 & 1
\end{pmatrix}
\;\; \Rightarrow \;\; T_0 = T_2 = \begin{pmatrix}
1
\end{pmatrix}
, \;\; T_1 = \begin{pmatrix}
0 & 1\\
1 & 0
\end{pmatrix}
\end{equation}
Here we can see that $T_1$ has in fact period $2$, which implies the presence of $2$ distinct limiting values whenever both $\rho$ and $H$ have a component belonging to $T_1$, consistently with what discussed above.

On the contrary, in accordance with \cref{thm:main}, the Ces\`aro average of the variance is always well-defined, and in this case we get

\begin{equation}
    \lim_{L\to\infty} \frac{1}{L}\sum_{l=1}^{L} \Varr^l = \frac{(\|\rho\|_2^2 -1/2)\|H\|_2^2}{6},
\end{equation}

which can be recovered both from direct calculation, and by computing the right leading eigenvector $w_1=(1/2, 1/2)^t$ of $T_1$.

\begin{figure}
    \includegraphics[width=0.8\linewidth]{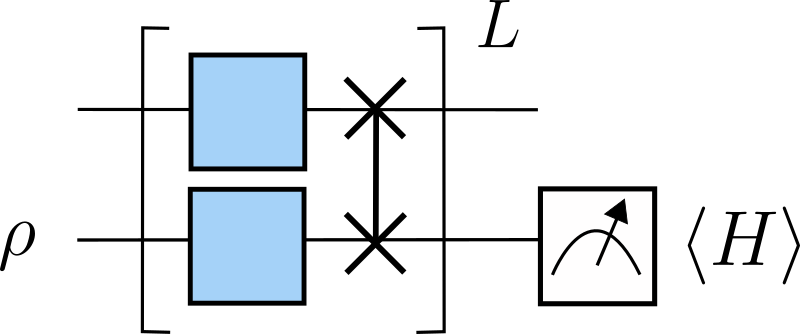}
    \caption{Simple 2-qubit circuit, designed to have a non-convergent variance. In this construction, the channel $\mathcal{E}$ is chosen to be unitary and in particular to be a SWAP gate. Moreover, both the initial state $\rho$ and the observable $H$ are chosen to be non-trivial only on the second qubit, and are therefore represented as single qubit operators.}
    \label{fig:non_conv_variance}
\end{figure}

%

\end{document}